\documentclass[conference]{IEEEtran}
\usepackage{amsfonts}
\usepackage{subfigure}
\usepackage{caption}

\usepackage{mdwlist}
\usepackage{cite}
\usepackage{amsfonts}
\usepackage{amssymb}
\usepackage{amsmath}          
\usepackage{amsthm}
\usepackage[usenames,dvipsnames]{xcolor}
\usepackage[ruled,vlined]{algorithm2e}
\usepackage{multirow}
\usepackage{bm}
\usepackage{booktabs}
\usepackage{graphics}
\usepackage{graphicx}
\usepackage{ifpdf}
\ifpdf
  \usepackage{epstopdf}
\fi

\usepackage{float}
\usepackage{epsfig}
\usepackage{multirow}
\usepackage{enumerate}

\usepackage{nomencl}
\makenomenclature

\usepackage{colortbl}  
\definecolor{tabcolor}{rgb}{1,0,0}  

\usepackage{threeparttable}
\usepackage{arydshln}

\usepackage{verbatim}
\usepackage{tikz,times}
\usetikzlibrary{calc, shapes, backgrounds}

\newtheorem{theorem}{Theorem}

\newtheorem{remark}{Remark}
\newtheorem{proposition}{Proposition}

\makeatletter
\newcommand\figcaption{\def\@captype{figure}\caption}
\makeatother
\ifCLASSINFOpdf
\else
\fi
\IEEEoverridecommandlockouts
\begin{document}

%
\title{Optimal Power Management for Failure Mode of MVDC Microgrids
in All-Electric Ships }
%

\author{~Qimin~Xu,~Bo~Yang, \IEEEmembership{Senior Member, IEEE},~Qiaoni~Han,~Yazhou~Yuan, \\
~Cailian~Chen,
\IEEEmembership{Member, IEEE},~Xinping~Guan, \IEEEmembership{Fellow, IEEE} \\

\thanks{This work was supported by National Key Research and Development Program of China (2016YFB090190), National Natural Science Foundation of China (61573245, 61521063, 61633017, 61622307, 61731012, 61803328 and 61803218). This work was also partially supported by SMC Outstanding Faculty Award of Shanghai Jiao Tong University.}
\thanks{Q.~Xu,~B.~Yang,~Q.~Han,~C.~Chen,~X.~Guan are with the Department of Automation, Shanghai Jiao Tong University, Shanghai 200240, China, Collaborative Innovation Center for Advanced Ship and Deep-Sea Exploration, Shanghai 200240, China, and also with the Key Laboratory of System Control and Information
Processing, Ministry of Education of China,
Shanghai 200240, China (e-mail: qiminxu@sjtu.edu.cn; bo.yang@sjtu.edu.cn; qiaoni@sjtu.edu.cn;
cailianchen@sjtu.edu.cn; xpguan@sjtu.edu.cn).}
\thanks{Y.~Yuan is with the Institute of Electrical Engineering,
Yanshan University, Qinhuangdao, 066004, China (e-mail: yzyuan@ysu.edu.cn).}
}


%


\maketitle

\begin{abstract}
Optimal power management of shipboard power system for failure mode (OPMSF) is a significant and challenging problem considering the safety of system and person. Many existing works focused on the transient-time recovery without consideration of the operating cost and the voyage plan.
In this paper, the OPMSF problem is formulated considering the mid-time scheduling and the faults at bus and generator. Two-side adjustment methods including the load shedding and the reconfiguration are coordinated for reducing the fault effects.
To address the formulated non-convex problem, the travel equality constraint and fractional energy efficiency operation indicator (EEOI) limitation are transformed into the convex forms. Then, considering the infeasibility scenario affected by faults, a further relaxation is adopted to formulate a new problem with feasibility guaranteed.
Furthermore, a sufficient condition is derived to ensure that the new problem has the same optimal solution as the original one.
Because of the mixed-integer nonlinear feature, an optimal management algorithm based on Benders decomposition (BD) is developed to solve the new one. Due to the slow convergence caused by the time-coupled constraints, a low-complexity near-optimal algorithm based on BD (LNBD) is proposed. The results verify the effectivity of the proposed methods and algorithms.
\end{abstract}


%
\IEEEpeerreviewmaketitle


\maketitle

%

\section{Introduction}
Shipboard power system (SPS) is self-powered by distributed electrical power generators operating collectively, which can be considered as an isolated microgrid. 
From the perspective of electrical design of all-electric ship (AES), there are three architectures of SPS to date, i.e., medium voltage DC (MVDC), medium voltage AC (MVAC), and higher frequency AC (HFAC). As the ever-increasing DC-based loads, it is likely that AES will feature a medium voltage primary distribution system in the future \cite{hebner2015technical}. 
Due to the intensive coupling and finite inertia feature, the consequences of a minor fault in a system component can be catastrophic.
Since the AES mostly targets at military applications, it is highly susceptible to be damaged.
Distinguished from terrestrial systems, the system failure of SPS is more disastrous due to the personnel safety on the shipboard.
Thus, optimal power management of SPS for failure mode (OPMSF) is essential to guarantee the system safety, while meeting the load demand.

\subsection{Motivation}
The time scale of the OPMSF problem includes transient-time, short-time, and mid-time.
Most of existing works about OPMSF focused on the recovery at a transient-time scale \cite{Bose2012Analysis,Srivastava2007Probability,das2013dynamic,jiang2012novel,Seenumani2012Real,Feng2015Multi} or a short-time scale \cite{seenumani2011reference}. Their objective is to improve the restored power of loads and guarantee the power balance. 
However, due to the damaged system structure by faults, the power supply-demand relationship is changed. From the results in \cite{Bose2012Analysis}, the delivered power falls to 75.4\% of the total power in 10\% of all possible 2-fault cases, and dips to 23.4\% in 10\% of all possible 3-fault cases. The imbalance of power supply and demand is severe in these cases. Consequently, the original optimal operating scheduling is not suitable for the remaining voyage, the operating cost and the risk of system safety are increased.
Hence, the mid-time scheduling OPMSF problem is essential and meaningful.

The adjustment methods for mid-time scheduling can be classified into two categories, i.e., the supply side and the demand side.
On the supply side, the generators that are the primary generation equipment cannot operate at original optimal state caused by faults.
Additionally, energy storage module compensation (ESMC) is a potential solution for improving the energy efficiency of SPS \cite{inventions2017Effect}.
Hence, the generation scheduling including ESMC in the remaining voyage has to be reorganized according to faults.
On the demand side, the load adjustment also plays a key role in optimal power management. The load of SPS includes the propulsion modules (PMs) and service loads. 
The propulsion power adjustment (PPA) can achieve energy efficiency improvement \cite{kanellos2014optimal}. However, once the capacity of generators is not enough to cover the load demand in the corresponding zones after faults happening, load shedding of service loads and reconfiguration of power network are required for guaranteeing the system safety and completing the voyage.
Hence, it is necessary to adopt load shedding and reconfiguration.
To sum up, two-side adjustment methods including load shedding and reconfiguration are meaningful and have significant effects on the OPMSF problem.



\subsection{Literature Review}
Many efforts have been devoted to study the fault management of SPS \cite{Bose2012Analysis,das2013dynamic,amba2009genetic,Nelson2015Automatic,Mitra2011Implementation,jiang2012novel,mashayekh2015integrated,Srivastava2007Probability,Seenumani2012Real,Feng2015Multi,kanellos2014optimal,seenumani2011reference,Kanellos2016Smart,Kanellos2017cost,Kanellos2014Optimaldemand,shang2016economic}. 
In the transient-time scale, their main objective is to maximize the weighted sum of restored loads \cite{Bose2012Analysis,das2013dynamic,jiang2012novel,Srivastava2007Probability,Seenumani2012Real,Feng2015Multi,amba2009genetic,Nelson2015Automatic,Mitra2011Implementation}. Additionally, there are other considerations including 
obtaining the correct order of switching\cite{das2013dynamic},  
probability-based prediction of fault effects\cite{Srivastava2007Probability}, 
minimizing the number or cost of switching actions\cite{Bose2012Analysis,jiang2012novel}, 
real-time management\cite{Seenumani2012Real,Feng2015Multi}, etc.
In the short-time scale, 
the authors in \cite{seenumani2011reference} developed a reference governor-based control approach to support the non-critical loads as much as possible while maximizing the battery usage.
However, they focused on the recovery of power supply for loads, without consideration of post-fault management including the voyage plan and the operating cost in a mid-time scale.

The optimal power management of SPS (OPMS) problem including the voyage plan and the operating cost in the mid-time scale has been studied in \cite{kanellos2014optimal,Kanellos2016Smart,Kanellos2017cost,Kanellos2014Optimaldemand,shang2016economic}. Particle swarm optimization (PSO) based algorithms are developed to solve the OPMS problem in \cite{kanellos2014optimal,Kanellos2016Smart,Kanellos2017cost}. In the three works, PPA and ESMC are considered for improving energy efficiency.
In \cite{Kanellos2014Optimaldemand}, dynamic programming (DP) algorithm is adopted to solve the same problem. In \cite{shang2016economic}, the authors formulated a multi-objective problem that considers the reduction of fuel consumption and energy efficiency operation indicator (EEOI) limitation together.
In the above works, they do not consider the failure mode.
To the best of our knowledge, there is no work focused on the OPMSF problem in a mid-time scale. Meanwhile, due to the damaged power network, and the usage of reconfiguration and load shedding, the algorithms in above works cannot be directly adopted to solve the OPMSF problem.

\subsection{Challenges}
The main target of this work is to solve the OPMSF problem in a mid-time scale. There are three main challenges to solve the problem.
Firstly, how to coordinate load shedding with the other adjustment methods to meet the load demand in the first place, unless faults affect the equipment safety and the voyage plan.
Secondly, considering the non-convex feature of the proposed problem and the infeasibility scenario affected by faults, it is hard to obtain the optimal solution and even a feasible solution.
Thirdly, the variables in the travel and ESM constraints of this mixed-integer nonlinear programming (MINLP) problem are coupled in time. 
Hence the computational complexity would be exponentially increasing with the number of operation time.


\subsection{Contributions} 
In this paper, the proposed OPMSF problem in MVDC SPS is to minimize the total operating cost including the cost of generation and energy storage while guaranteeing the system safety, the GHG emission limitation, and the voyage plan.
The contributions of this paper are summarized below.
\begin{itemize}
\item {
The OPMSF problem is reformulated based on the analysis of faults.
Load shedding and reconfiguration are added as auxiliary adjustment methods considering the fault effects. 
Different from existing works, to guarantee that load shedding only works when generator scheduling (GS) and ESMC cannot solve the OPMSF problem, 
a coordination mechanism is developed by adding a penalty term of load shedding in the objective and a sufficient condition of the penalty parameter is derived. 
}
\item { { Non-convex travel constraint and fractional EEOI limitation are transformed into convex forms to obtain a better tractable problem.
Then, considering the infeasible scenarios caused by faults, a feasibility-guarantee mechanism is established by introducing a slack distance variable and adding its penalty term in the objective.
Lastly, a sufficient condition of that penalty parameter is derived to guarantee that if the original problem is feasible, the new one has the same optimal solution; if not, the maximum travel distance can be further obtained to assist rescue mission. }}
\item {To address the reformulated problem, {an {optimal management algorithm} based on Benders decomposition (BD) is designed to split it into two more tractable problems (subproblem and master problem)}. 
Due to the slow convergence caused by the time-coupled constraints,
{a low-complexity near-optimal algorithm based on BD (LNBD) is proposed by decomposing the time-coupled constraints with suboptimal power allocation of ESMs and propulsion modules in the subproblem. A complexity analysis is given to compare the performance of two algorithms.
}
}
\end{itemize}


The paper is organized as follows: in Section \ref{sec:model}, the main modules are introduced, and the OPMS problem is formulated; Section \ref{sec:problem_transformation} reformulates the OPMSF problem according to different faults; Section \ref{sec:algorithm_design} details the proposed algorithms; the performance of the proposed algorithms are evaluated in Section \ref{sec:simulation}. Finally, the conclusion is drawn in Section \ref{sec:conclusion}.

\section{System Models and Problem Formulation}
\label{sec:model}
SPS is an integrated power system, which consists of power generators, energy storage modules (ESMs), service loads,  propulsion modules (PMs), converters, and power network. 
In this section, these main models are introduced.
Then, the OPMS problem is formulated based on these models.
\begin{figure*}[ht]
\vspace{-2ex}
  \centering
    \includegraphics[width= 0.6 \textwidth]{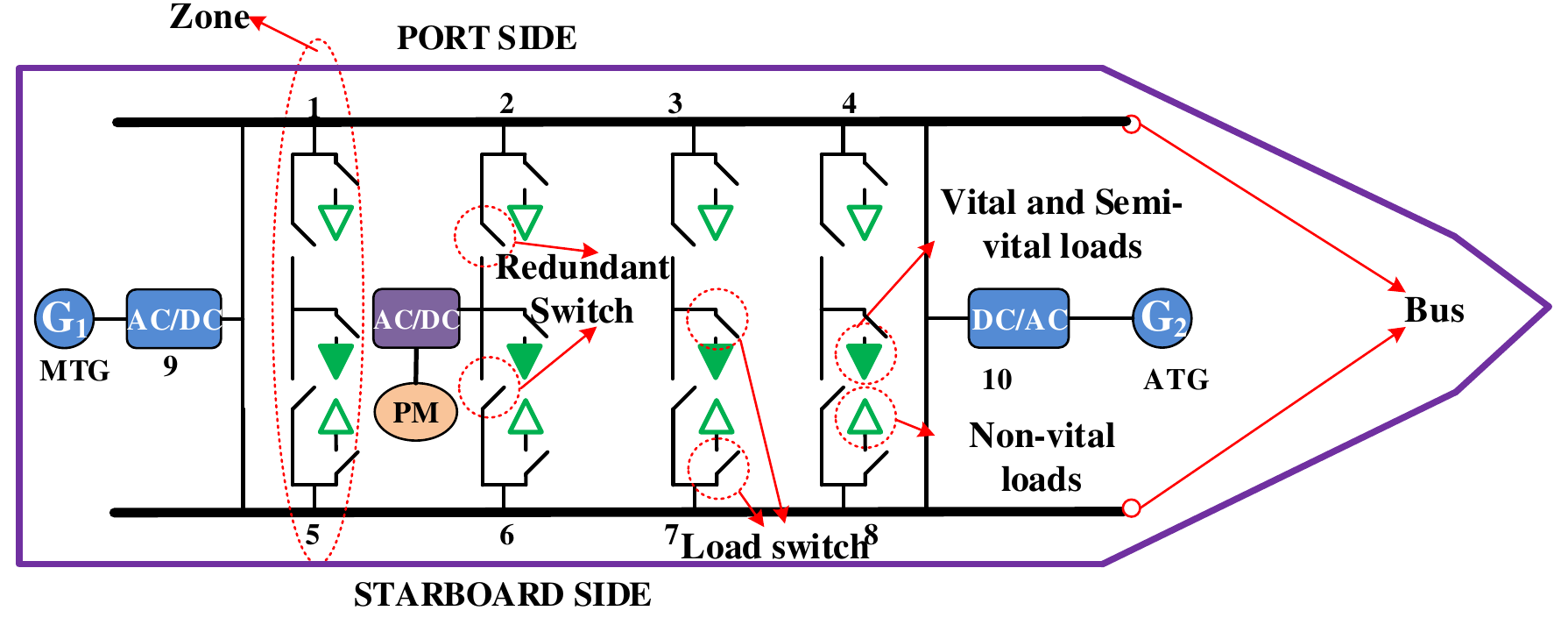}
  \caption{MVDC shipboard power system architecture.}
  \label{fig:arch}
  \vspace{-4ex}
\end{figure*}

\vspace{-2ex}
\subsection{System Structure Overview}
The classic architecture of MVDC SPS is shown in Fig.\ref{fig:arch}. This architecture adopts a zonal approach with a starboard bus (SB) and a port bus (PB), and the SPS is parted into $ Z $ electric zones. The generators are classified into two types: main turbine generator (MTG) and auxiliary turbine generator (ATG). 
The DC zones are powered by a set of generators and converters denoted by $m \in \mathcal{M}$.  The loads are powered by a set of buses which run longitudinally along the PB and SB. 
In this work, we assume that the SPS operates in discrete time with $t \in \mathcal{T} = \{1,2,3, \cdots \} $ and time interval $\Delta t$.

\begin{table}
\centering
  \caption{{Main notations and abbreviations}}
  \label{tab:1}
  \footnotesize
  \begin{tabular}{p{1.5cm}|p{6.5cm}} 
  \hline
  Notation    & Physical interpretation \\
  \hline
  \hline
  $ \mathcal{M} $, $ m $ & Set and index of generators and converters \\
  $ \mathcal{Z} $, $ z $ & Set and index of DC zones \\
  $ \mathcal{R} $, $ r $ & Set and indexes of propulsion modules \\
  $ \mathcal{N} $, $ n $ & Set and index of ESMs \\
  $ \mathcal{W} $, $W$, $ w $ & Set, number and index of island parts \\
  $ \mathcal{T} $, $T$, $ t $ & Set, number and index of time \\
  $ {\max} $, $ {\min} $ & Superscript denoting minimum and maximum \\
  $ P_{{\rm{g}}, m}(t)$ & Output power of generator $m$ at time $t$\\
  $ \delta_{{\rm{g}},m}(t) $  & Status of generator $m$, (1/0 = online/offline)\\
  $ P_{{\rm{e}}, n}(t)$  & Output power of the ESM in zone $z$ at time $t$ \\
  $ E_{{\rm{e}}, n}(t)$   & Capacity of the ESM in zone $z$ at time $t$\\
  $ P_{{\rm{pr}}, r}(t) $ & Propulsion power of the $ r $-th propulsion module\\
  $ P_{\rm{vs}}(t) $  & Power demand of vital and semi-vital loads at time $t$ \\
  $ P_{\rm{no}}(t) $ & Power demand of non-vital loads at time $t$ \\
  $ \rho(t) $ & Lower bound of the load-shedding amount of $P_{\rm{no}}(t)$\\
  $ P_{{\rm{G}}, w}(t) $  & Output power of generators in island part $w$ at time $t$\\
  $ P_{{\rm{E}}, w}(t) $  & Output power of the ESMs in island part $w$ at time $t$\\
  $ P_{{\rm{PR}}, w}(t) $  & Output power of the propulsion modules in island part $w$ at time $t$\\
  $ P_{{\rm{L}}, w}(t) $  & Total load demand in island part $w$ at time $t$\\
  $ V(t) $ & Ship speed at time $t$\\
  $ C(t) $ & Total operating cost at time $t$\\
  $ C_{\rm{G}}(P_{{\rm{g}},m}(t)) $ & Fuel cost of generator $m$ at time $t$\\ 
  $ C_{\rm{E}}(P_{{\rm{e}}, n}(t)) $ & Operating cost of the ESM $z$ at time $t$\\  
  $ C_{\rm{L}}(\rho(t)) $ & Load shedding cost at time $t$ \\
  $ S_{{\rm{P}}, z}(t) $ & Redundancy switches of PB in zone $ z $ at time $t$\\
  $ S_{{\rm{S}}, z}(t) $ & Redundancy switches of SB in zone $ z $ at time $t$\\
  $ D, D_{\rm{d}} $  & Travel distance, reduced travel distance \\ \hline
  \end{tabular}
\end{table}

\vspace{-2ex}
\subsection{Generation Model}
The output power has to be bounded to keep the safe operation of generators and avoid mechanical damage. Those constraints are described as follows, $\forall t \in \mathcal{T}, \forall m \in \mathcal{M}$:
\begin{align} 
\label{eqn:PG_limit} & \delta_{{\rm{g}},m}(t) P_{{\rm{g}}, m}^{\min} \leqslant P_{{\rm{g}}, m}(t) \leqslant  \delta_{{\rm{g}},m}(t) P_{{\rm{g}}, m}^{\max}, \\ 
\label{eqn:P_Rate_limit} & -R_{{\rm{g}},m}^{\max} \leqslant {  P_{{\rm{g}}, m}(t)} - P_{{\rm{g}}, m}(t-1) \leqslant R_{{\rm{g}},m}^{\max}, \\
\label{eqn:start-up_detect} & \delta_{{\rm{g}},m}(t) - \delta_{{\rm{g}},m}(t-1) \leqslant y_{{\rm{g}},m}(t),\\
\label{eqn:min_time_operation} & y_{{\rm{g}},m}(t)T_{m}^{\min} \leqslant \delta_{{\rm{g}},m}(t) +  \cdots + \delta_{{\rm{g}},m}(t+T_{m}^{\min}-1) ,
\end{align}
where $P_{{\rm{g}}, m}(t)$ denotes the output power of generator $m$ at time $t$,
$\delta_{{\rm{g}},m}(t) $ the status of generator $m$,
$R_{{\rm{g}},m}^{\max}$ the maximum ramp-rate of generator $m$. 
The ramp-rate of $P_{{\rm{g}},m}(t)$ is limited by (\ref{eqn:P_Rate_limit}). The start-up state $ y_{{\rm{g}},m}(t) $ (binary variable) is detected by (\ref{eqn:start-up_detect}). Eq. (\ref{eqn:min_time_operation}) describes operation time management where $T_m^{\min}$ denotes the minimum operation time of generator $m$.

The fuel consumption cost $C_{\rm{G}}(P_{{\rm{g}}, m}(t))$ of generator $m$ at time $t$ is expressed as, $\forall m \in \mathcal{M} $:
\begin{align*}
C_{\rm{G}}(P_{{\rm{g}}, m}(t)) &  =  
a_{{\rm{g}},m} ( P_{{\rm{g}}, m}(t))^2 \Delta t + b_{{\rm{g}},m}  P_{{\rm{g}}, m}(t) \Delta t \\ 
&  + c_{{\rm{g}},m} \delta_{{\rm{g}},m}(t) \Delta t  , 
\end{align*}
where $C_{\rm{G}}(\cdot)$ denotes the fuel cost function in arbitrary monetary unit (m.u.), which can be approximately represented by a quadratic function of produced power $P_{{\rm{g}}, m}(t)$.
$a_{{\rm{g}},m}$, $b_{{\rm{g}},m}$, and $c_{{\rm{g}},m} $ are constants determined by technical specifications of generator $m$. 
{Generators in SPS are too small in size to heat up the equipments in hours to drive the steam turbine compared to that in terrestrial grids. The maximum startup time is typical five minutes \cite{ieeestd4512017}. Thus the startup cost is neglected.}


\subsection{Energy Storage Module }
\label{subsec:ESM}
The third option of ESM location of multi-zone SPS in \cite{yan2011optimal} is employed in each zone, which is shown in Fig. \ref{fig:ESM}.
ESM is incorporated with the PB or SB at the longitude bus. 
In this option, ESM can supply power for propulsion modules and service loads. 
 \begin{figure}[ht]
 \begin{center}
 \setlength{\belowcaptionskip}{-0.5cm}
 \includegraphics[width= 0.75\linewidth]{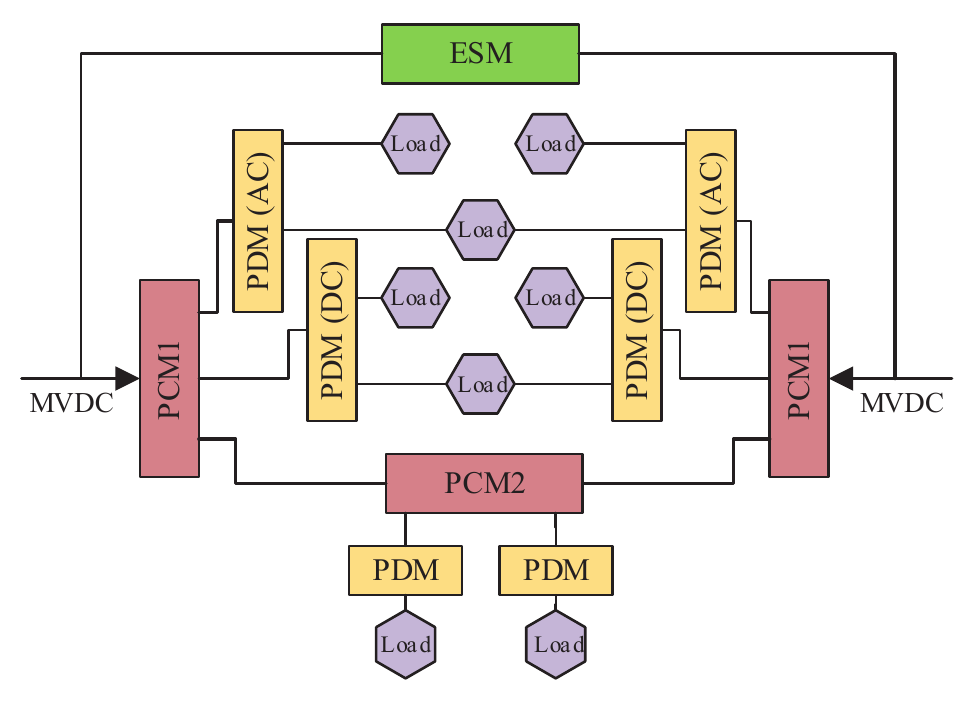}
 \caption{ Energy storage module location in zonal power system.} 
 \captionsetup{justification=centering} 
 \label{fig:ESM} 
 \end{center}
 \end{figure}
The maximum produced or absorbed power of ESMs are denoted by the maximum charging and discharging powers $P_{\rm{e}}^{\max}$, $P_{\rm{e}}^{\min}$, respectively. Hence the ESMs satisfy the following constraints, {$\forall t \in \mathcal{T}, \forall n \in \mathcal{N} \subseteq \mathcal{Z}$}:
\begin{align}
\label{eqn:pe_range} & P_{\rm{e}}^{\min} \leqslant P_{{\rm{e}}, n}(t) \leqslant  P_{\rm{e}}^{\max}, \\
\label{eqn:e_capacity_range} & E^{\min} \leqslant E_{{\rm{e}}, n}(t) \leqslant E^{\max}, \\
\label{eqn:e_capacity_change} & E_{{\rm{e}}, n}(t) = E_{{\rm{e}}, n}(t-1) + P_{{\rm{e}}, n}(t) \Delta t,
\end{align}
where $E_{{\rm{e}}, n}(t)$ and $P_{{\rm{e}}, n}(t)$ represent residual energy and output power of the ESM in zone $z$ at time $t$. $P_{\rm{e}}^{\min}$ is a negative constant. 
{$E^{\min}$ and $E^{\max}$ are the lower and upper bounds of $E_{{\rm{e}},n}$.}
The capital cost is constant, because it is just relative to $E^{\max}$ and $P^{\max}$ \cite{shi2015real,schoenung2011energy,Khani2015Real}. Besides,
the charge energy cannot be purchased from the main grid on the voyage. {Thus, the capital cost and power purchasing cost of ESM are ignored here. 
The life-cycle cost $ C_{\rm{lc}}(P_{{\rm{e}}, n}(t)) $ of ESM $n$ at each time $t$ is mainly considered, which is modeled as \cite{shi2015real},} $ \forall n \in \mathcal{N}$:
\begin{equation*}
\begin{aligned}
\label{eqn:battery_cost}
{ C_{\rm{E}}(P_{{\rm{e}}, n}(t)) = a_{\rm{lc}} P_{{\rm{e}}, n}(t)^2 \Delta t + c_{\rm{lc}} \Delta t,}
\end{aligned}
\end{equation*}
where $ a_{\rm{lc}}$ and $ c_{\rm{lc}} $ are positive constants.

\vspace{-2ex}
\subsection{Service Loads}
In SPS, all the service loads can be divided into three categories: vital, semi-vital, and non-vital loads.
Vital loads cannot be interrupted, which are always required for normal mode. Semi-vital loads can be interrupted in a short-time scale. In the case of emergencies, non-vital loads can be shed in a mid-time scale to maintain power balance. At a certain time, vital and semi-vital loads can be powered by the PB or SB using redundant switches which are represented by a pair of switches as shown in Fig.\ref{fig:arch}. 
The constraints are detailed in \ref{subsec:zone_switch}. Each non-vital load connects to one bus, PB or SB. 
{ Each load is controlled by a switch that is used for load shedding.}
In this work, load shedding is only considered for non-vital loads. 
Thus the demand at time $t$ can be described as:
\begin{align}
\label{eqn:total_load} P_{\rm{L}}(t) = P_{\rm{vs}}(t) + (1 - \rho (t))P_{\rm{no}}(t), \rho (t) \in [0,1].
\end{align}
where $\rho(t)$ indicates the lower bound of the load-shedding amount of all the non-vital loads at time $t$, which is a continuous variable to reduce the computation complexity.

\vspace{-3ex}
\subsection{Propulsion Module}
The ship velocity is determined by the propulsion power.
The relationship between them depends on hull resistance at specific conditions and is formulated as:
\begin{equation}
\label{eqn:prop_speed}
P_{\rm{PR}}(t) = \alpha  V(t)^{\beta}, \quad \forall t \in \mathcal{T},
\end{equation}
where $V(t)$ denotes the ship velocity at time $t$, $P_{\rm{PR}}(t)$ the total required propulsion power to reach the velocity $V(t)$, $\alpha$ the matching coefficient for propulsion power and velocity, and $\beta$ a constant which equals to 3 for conventional hull form \cite{kanellos2014optimal}. Propulsion power is the sum of all the propulsion power $P_{\rm{PR}}(t) = \sum_{r \in \mathcal{R} } P_{{\rm{pr}},r}(t)$.
Here the velocity is bounded by the maximum and minimum ship speed. 
\begin{align} 
\label{eqn:speed}& V^{\min} \leqslant V(t) \leqslant  V^{\max}, \quad \forall t \in \mathcal{T}.
\end{align}

The total travel distance of all the time should almost equal to the travel distance target $D$. 
\begin{align} 
\label{eqn:rest_distance}&  \sum\nolimits_{t \in \mathcal{T}} V(t) \Delta t = D .
\end{align}

\vspace{-4ex}
\subsection{Power Network Model}
\label{subsec:zone_switch}
\subsubsection{Zone Redundant Switch}
The redundant design of zonal SPS is used for improving the reliability of power supply for vital and semi-vital loads. The structure diagram is shown in Fig. \ref{fig:arch}.
In each zone, every vital and semi-vital loads can be powered by PB or SB at a certain time. The redundant switches determine that the vital and semi-vital loads are powered by PB or SB. $ S_{{\rm{P}}, z}(t)$ and $S_{{\rm{S}}, z}(t) $ ( 0/1=open/close ) denote the redundant switches connected with PB and SB in zone $z$, respectively. 
Thus if ${S}_{{\rm{P}}, z}(t) = 1$ and $ {S}_{{\rm{S}}, z}(t) = 0$, the vital and semi-vital loads in zone $z$ are powered by PB, and conversely powered by SB. 
Thus, the related constraints are written as, $\forall t \in \mathcal{T},  \forall z \in \mathcal{Z}$:
\begin{align} 
\label{eqn:redundant_switches} & {S}_{{\rm{P}}, z}(t) + {S}_{{\rm{S}}, z}(t) =1, \ {S}_{{\rm{P}}, z}(t), {S}_{{\rm{S}}, z}(t) \in \{0,1\}, \\
\label{eqn:redundant_switches_time1} & y_{{\rm{P}},z}(t)T_{s}^{\min} \leqslant {S}_{{\rm{P}}, z}(t) +  \cdots + {S}_{{\rm{P}}, z}(t+T_{s}^{\min}-1), 
\end{align}
where constraint (\ref{eqn:redundant_switches_time1}) describes the minimum switching time $T_{s}^{\min}$ to avoid frequently switching.

\subsubsection{Power Balance}
Each generator directly connects to the corresponding converter.
The power equation of converter $m$ is described as:
\begin{align} 
\label{eqn:ACDC_output_load} &  P_{{\rm{oc}}, m}(t) = \zeta  P_{{\rm{ic}}, m}(t),  \quad \forall m \in \mathcal{M}, t \in \mathcal{T},
\end{align}
where $P_{{\rm{ic}}, m}(t)$ and $P_{{\rm{oc}}, m}(t)$ represent the input and output power of converter $m$. Here the power loss of converter is considered as a constant ratio $(1 - \zeta)$ with the input power $P_{{\rm{ic}},m}(t)$. 
{Since the generator and converter are tight coupled, $P_{{\rm{g}},m}(t) \approx P_{{\rm{ic}}, m}(t)$.} Thus (\ref{eqn:ACDC_output_load}) can be transfered into
\begin{align} 
\label{eqn:ACDC_relationship} &  P_{{\rm{oc}}, m}(t) = \zeta P_{{\rm{g}},m}(t), \quad \forall m \in \mathcal{M}, t \in \mathcal{T}.
\end{align}

The power in DC part comes from converters and ESMs. The loads in DC part include service loads, propulsion modules. Due to the tight couple in SPS, the power loss of transmission line can be ignored.
Hence the supply and demand balance in DC part is given as, $\forall t \in \mathcal{T}$:
\begin{equation}
\begin{aligned} 
\label{eqn:DC_output_load}  \sum\nolimits_{m \in \mathcal{M}} P_{{\rm{oc}}, m}(t) + \sum\nolimits_{n \in \mathcal{N}} P_{{\rm{e}}, n}(t) 
= P_{\rm{L}}(t) + P_{{\rm{PR}}}(t).
\end{aligned}
\end{equation}

\subsection{Greenhouse Gas Emissions}
According to the International Maritime Organization policy, there are two indicators: Energy Efficiency Design Indicator (EEDI) and Energy Efficiency Operation Indicator (EEOI). 
Considering that only one operation point is evaluated in EEDI, the limitation of GHG emissions cannot be guaranteed during the entire operation. Hence, EEOI is more suitable for the GHG emission evaluation in the entire shipboard operation. 
EEOI is defined as the ratio between produced $\rm{CO}_2$ mass and transport work \cite{Guidelines2009IMO}. Thus, EEOI limitation is represented by, $\forall t \in \mathcal{T}$:
\begin{equation}
\label{eqn:eeoi_limit}
\begin{aligned}
\dfrac{\sum_{ g \in \mathcal{G} } F{(P_{\rm{g}, m}(t))}}{F_{\rm{sl}} V(t) \Delta t} \leqslant \rm{EEOI}^{\max},
\end{aligned} 
\end{equation}
where $F(P_{\rm{g}, m}(t)) $ is the function of produced $\rm{CO_2}$ mass.
$\rm{EEOI}^{\max}$ is the EEOI limitation that is in $\rm{gCO_2tn^{-1}nm^{-1}}$. 
The ship load factor $F_{\rm{sl}}$ is determined by ship type and carried cargo, which is in tonne (tn) or kilotonne (ktn).

According to \cite{shang2016economic}, the function $F(P_{\rm{g}, m}(t)) $ is expressed as, $\forall t \in \mathcal{T}, g \in \mathcal{G} $: 
{\small
{\begin{equation*} 
F(P_{\rm{g}, m}(t)) \stackrel{\vartriangle}{=} a_{\rm{g}, m} ( P_{\rm{g}, m}(t))^2 \Delta t + b_{\rm{g}, m} P_{\rm{g}, m}(t) \Delta t + c_{\rm{g}, m} \delta_{g}(t) \Delta t.
\end{equation*} }}

\vspace{-2ex}
\subsection{ Optimal Power Management of SPS Problem}
To address the OPMSF problem, it needs to formulate the OPMS problem at first. 
$S_{{\rm{P}},z}(t)$ and $S_{{\rm{S}},z}(t)$ do not need to be reconfigured in the normal mode. Thus, they are not included in the control vector $\bm u(t)$ that is defined as
\begin{equation*}
\bm u(t) \stackrel{\vartriangle}{=} ( \bm \delta_{{\rm{g}},m}(t), \bm y_{{\rm{g}},m}(t), \bm P_{{\rm{g}},m}(t), \bm P_{{\rm{e}}, n}(t), \bm P_{{\rm{pr}},r}(t), \rho(t)).
\end{equation*}

Then, the objective is to minimize the total operating cost, which is defined as
\begin{equation*}
\begin{aligned}
C(t) & {\stackrel{\vartriangle}{=} \sum\nolimits_{m \in \mathcal{M}} C_{\rm{G}}(P_{{\rm{g}},m}(t)) + \xi_{\rm{e}} \sum\nolimits_{n \in \mathcal{N}} C_{\rm{E}}(P_{{\rm{e}}, n}(t) ) } \\
& {+ \xi_{\rm{l}} C_{\rm{L}}(\rho(t)), \forall t \in \mathcal{T},}
\end{aligned}
\end{equation*}
where 
$\xi_{\rm{e}}$ and $\xi_{\rm{l}}$ are the parameters used to make a tradeoff between the fuel consumption cost, the operating cost of ESMs, and the load shedding cost.
The cost functions of generators and ESMs are formulated according to the operating cost in operation. The cost function of load shedding is used to avoid shedding all the service loads for lower operating cost, which is a penalty cost for the coordination with other adjustment methods. The cost function of load shedding $C_{\rm{L}}(\rho(t))$ at each time $t$ is represented by the load-shedding amount $P_{\rm{ls}}(t)$.
\begin{equation*}
C_{\rm{L}}(\rho(t)) = P_{\rm{ls}}(t) = \rho(t) P_{\rm{no}}(t) \Delta t.
\end{equation*}

Hence the OPMS problem integrated with load shedding 
is described as:
\begin{equation*}
\begin{aligned}
\label{Objective} \textbf{P1}:  \;
\underset{ \bm u(t) } \min \; & \sum\nolimits_{t \in \mathcal{T}} C(t) \\
\text{s.t.} \quad & (\ref{eqn:PG_limit})-(\ref{eqn:redundant_switches_time1}), (\ref{eqn:ACDC_relationship})-(\ref{eqn:eeoi_limit}),
\end{aligned}
\end{equation*}

\begin{proposition}
For meeting the load demand firstly, load shedding is only adopted when GS and ESMC cannot solve the OPMSF problem. Thus, $\xi_{\rm{l}}$ has to satisfy that
\begin{equation}
\label{eqn:load_sheddding_penalty_xi}
\xi_{\rm{l}} >  \max \{ 2 a_{g} P_{g}^{\max} + b_{g}, 2 \xi_{\rm{e}} a_{\rm{lc}} P_{\rm{e}}^{\max}  \},
\end{equation}
where $ P_{g}^{\max} = \underset{m} \max \left( P_{{\rm{g}},m}^{\max} \right), a_{g} = \underset{m} \max \left( a_{{\rm{g}},m} \right), b_{g} = \underset{m} \max \left( b_{{\rm{g}},m} \right), m \in \mathcal{M}$, $P_{\rm{e}}^{\max} = \underset{n} \max \left( P_{{\rm{e}}, n}^{\max} \right), n \in \mathcal{N} $.
\end{proposition}
\begin{proof}
Please see Appendix \ref{append:theorem1}.
\end{proof}

 
\section{Problem Analysis and Transformation}
\label{sec:problem_transformation}
{
In this section, fault preprocessing and problem transformation are carried out to make the problem more tractable.} The problem transformation is illustrated in Fig. \ref{fig:Problem_transformation}.
\begin{figure}[ht]
\begin{center}

\setlength{\belowcaptionskip}{-0.5cm}
\includegraphics[width= 0.99 \linewidth]{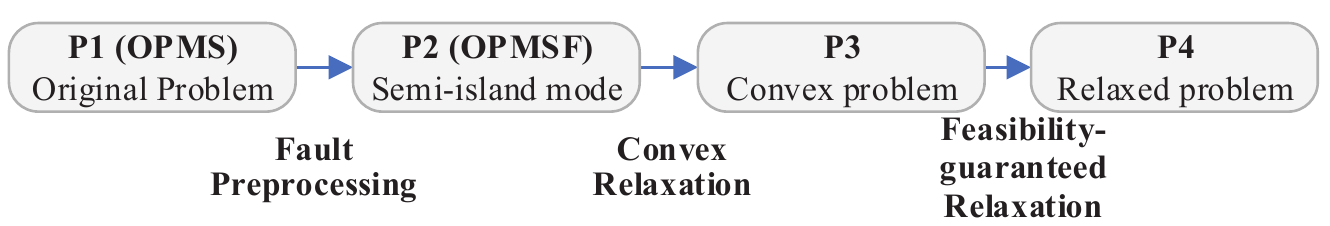} 
\caption{Problem transformation.}
\captionsetup{justification=centering} 
\label{fig:Problem_transformation}
\end{center}  
\end{figure} 
%
%

%
\subsection{Fault Preprocessing}
{
This part includes fault analysis and problem reformulation.}

\subsubsection{Fault Analysis}
In this work, two main physical faults are considered: generator fault and bus fault. Based on our analysis,
all the faults are divided into three modes, i.e., island fault, semi-island fault, and non-island fault. The examples of latter two are shown in Fig. \ref{fig:island}.

\begin{figure}[ht]
\begin{center} 
\setlength{\belowcaptionskip}{-0.3cm}
\includegraphics[width= 0.48 \textwidth]{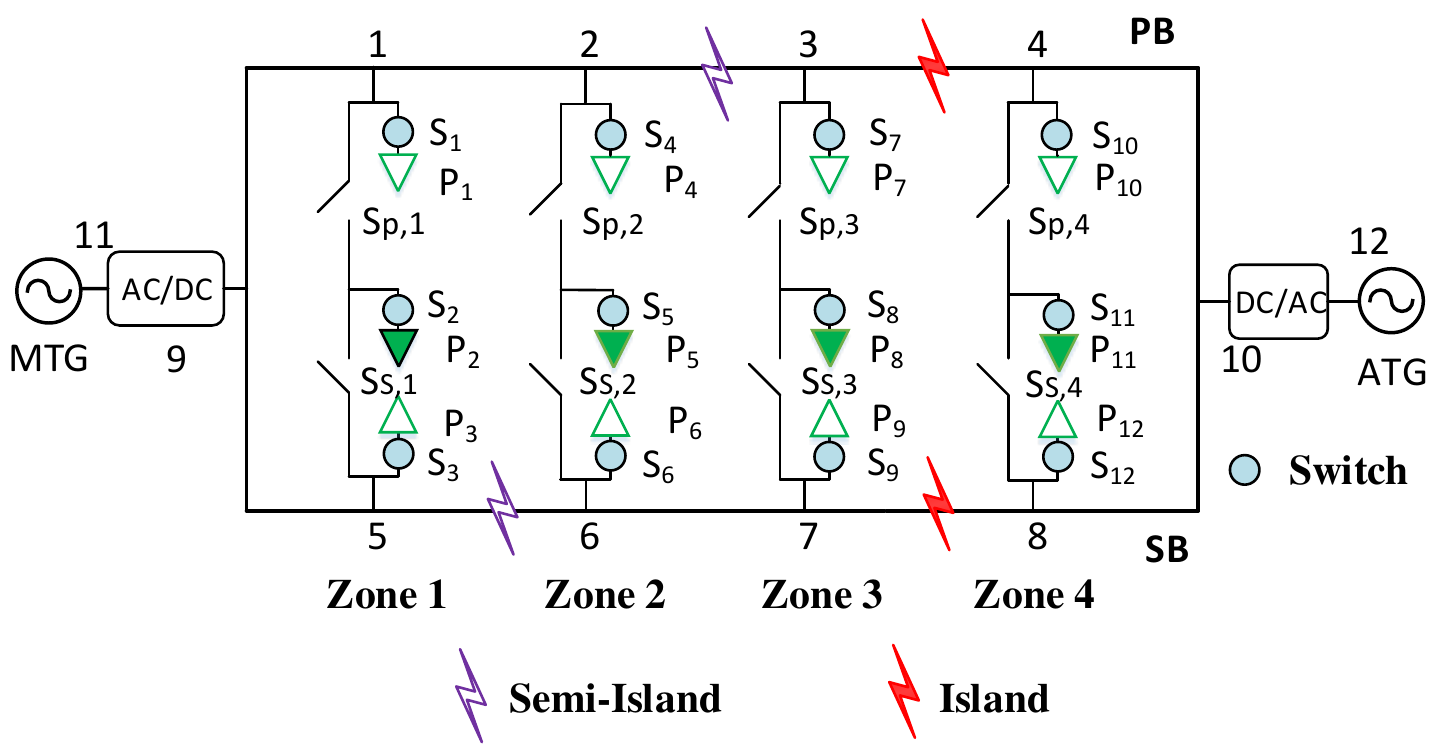} 
\caption{ Semi-island and island scenario. }
\captionsetup{justification=centering} 
\label{fig:island}
\end{center}
\end{figure}

\paragraph{\textbf{Non-island mode}} If faults only happen at one bus (PB or SB), these cases are defined as the non-island mode. In this mode, the system structure has not been significantly changed, and the power balance constraint is also not affected.

\paragraph{\textbf{Island mode}} Faults that happen at both sides (PB and SB) divide the zones into several island parts. 
In this mode, the redundant switches do not need to be changed. A case is shown in Fig. \ref{fig:island}. The DC zones is divided into two parts without any connection: zone 1, 2, 3 and zone 4.

\paragraph{\textbf{Semi-island mode}} This mode is similar to the island mode, but there are coupled zones between the island parts. For example, in the semi-island mode of Fig. \ref{fig:island}, the loads in zone 2 can be powered by two islanding parts. If $S_{{\rm{S}},2} = 1$ and $S_{{\rm{P}},2} = 0$, zone 2 is powered by ATG, and conversely powered by MTG. Zone 2 is the coupled part, and $S_{{\rm{S}},2}$, $S_{{\rm{P}},2}$ are the coupled redundant switches.

\subsubsection{ Problem reformulation }
Based on three fault-modes, the OPMS problem is reformulated respectively.

\paragraph{\textbf{Non-island Mode}}
In the non-island mode, all the service loads and propulsion modules can be powered by MTGs and ATGs, and the redundant switches are reconfigured to connect to the undamaged bus (PB or SB). Thus, the OPMS problem for this fault mode is \textbf{P1} with determined redundant switch configuration. Similarly, the OPMS problem in generator fault is \textbf{P1} with determined state $\delta_{{\rm{g}},m}(t)$. Thus, the optimization problems in this mode are classified as \textbf{P1}.

\paragraph{\textbf{Island Mode}}
Since island-mode faults damage the power network, the power balance constraint (\ref{eqn:DC_output_load}) are correspondingly changed. 
Each island part is denoted by $w \in \mathcal{W}$. Here the power constraint are reconstructed as, $\forall t \in \mathcal{T}, \forall w \in \mathcal{W}$:
{\small
\begin{align}
\label{eqn:DC_output_load_fault_detail} 
 & \sum\nolimits_{z \in \mathcal{L}_{w}} P_{{\rm{vs}},z}(t) + (1-\rho_w(t))P_{{\rm{NO}},w}(t) + \sum\nolimits_{r \in \mathcal{R}_{w} } P_{{\rm{pr}}, r}(t) \nonumber \\
 & = \sum\nolimits_{m \in \mathcal{M}_{w}} \zeta P_{{\rm{g}}, m}(t) + \sum\nolimits_{n \in \mathcal{N}_{w}} P_{{\rm{e}}, n}(t) ,
\end{align}
where }$ \mathcal{L}_{w} $, $ \mathcal{R}_{w} $, $ \mathcal{M}_{w} $, and $ \mathcal{N}_{w} $ denote the sets of service loads (vital and semi-vital), propulsion modules, generators, and ESMs in the $w$-th island part respectively.
Hence the constraint (\ref{eqn:DC_output_load}) is divided into $W$ constraints. $P_{{\rm{vs}},z}(t)$ denotes the vital and semi-vital loads in zone $z$ at time $t$.
{$P_{{\rm{NO}}, w}(t)$ represents non-vital loads in island $w$}. 
$P_{{\rm{G}}, w}(t)$ and $P_{{\rm{E}}, w}(t)$ denote the output power of the generators and ESMs in the $w$-th island part, respectively.
$P_{{\rm{L}}, w}(t)$ and $P_{{\rm{PR}}, w}(t)$ are the demand of the service loads and propulsion modules in the $w$-th island part, respectively. 
Then, constraint (\ref{eqn:DC_output_load_fault_detail}) can be simplified as, $\forall t \in \mathcal{T}, \forall w \in \mathcal{W}$:
\begin{equation}
\begin{aligned}
\label{eqn:DC_output_load_fault} & P_{{\rm{L}}, w}(t) + P_{{\rm{PR}}, w}(t) = \zeta P_{{\rm{G}}, w}(t) + P_{{\rm{E}}, w}(t),
\end{aligned}
\end{equation}
where $ P_{{\rm{L}}, w}(t) = P_{{\rm{VS}}, w}(t) + (1-\rho_w(t))P_{{\rm{NO}}, w}(t) $. $P_{{\rm{VS}}, w}(t)$ and $P_{{\rm{NO}}, w}(t)$ denote the vital and semi-vital loads, and non-vital loads in island $w$ respectively.

Since the system is divided into $W$ island parts and the propulsion modules in one zone practically, \textbf{P1} is divided into $W$ independent problems. 
Only one problem needs to adjust $P_{{\rm{PR}},w}$. 
Based on the relationship of speed and propulsion power in (\ref{eqn:prop_speed}), the speed related constraints (\ref{eqn:speed}) and (\ref{eqn:rest_distance}) can be transformed as follows, $\forall t \in \mathcal{T}$:
{{\begin{align}
\label{eqn:distance_ppr}
& \sum\nolimits_{t \in \mathcal{T}} ( P_{\rm{PR}}{(t)} / \alpha )^{1/\beta} \Delta t = {D} ,\\
\label{eqn:range_ppr1} 
& \alpha ( V^{\min} )^{\beta} \leqslant P_{\rm{PR}}{(t)} \leqslant \alpha ( V^{\max} )^{\beta}.
\end{align}}
}

Thus the control vector of each part in island mode at time $t$ can be represented by 
{
{\begin{equation*}
\begin{aligned}
{\bm u}_w(t) \stackrel{\vartriangle}{=} & (\bm \delta_{{\rm{g}},m}(t), \bm y_{{\rm{g}},m}(t), \bm P_{{\rm{g}},m}(t), \bm P_{{\rm{e}}, n}(t), \bm P_{{\rm{pr}},r}(t), \rho_w (t)), \\
& \;  m \in \mathcal{M}_{w}, n \in \mathcal{N}_{w}, r \in \mathcal{R}_{w}.
\end{aligned}
\end{equation*}}
}

The objective function of OPMS problem for this fault mode in each island part $w$ is defined as 
{
{\begin{equation*}
\begin{aligned}
C_{w}(t) &  \stackrel{\vartriangle}{=} \sum\nolimits_{m \in \mathcal{M}_w } C_{\rm{G}}(P_{{\rm{g}},m}(t)) + \xi_{\rm{e}} \sum\nolimits_{z \in \mathcal{N}_w} C_{\rm{E}}(P_{{\rm{e}}, n}(t)) \\
&  + \xi_{\rm{l}} C_{\rm{L}}(\rho_{w}(t)), \forall t \in \mathcal{T} .
\end{aligned}
\end{equation*}}
}

\paragraph{\textbf{Semi-island Mode}}
In semi-island mode, the constraint (\ref{eqn:DC_output_load}) is also divided into $W$ constraints. 
The coupled redundant switches $S_{{\rm{P}},x}, x \in \Omega_{{\rm{pb}},w} $ and $S_{{\rm{S}},y}, y \in \Omega_{{\rm{sb}},w} $ affect the total power of loads of this mode. $ \Omega_{{\rm{pb}},w} $ and $ \Omega_{{\rm{sb}},w} $ denote the sets of coupled redundant switches connected to PB and SB respectively, which belong to the $w$-th island part.
In other words, the coupled redundant switch reconfiguration is to determine which island part supplies vital and semi-vital loads in the coupled zones. Each power balance constraint is related to the $S_{{\rm{P}},x}$ and $S_{{\rm{S}},y}$ in the coupled zones, which can be described as, $\forall t \in \mathcal{T}, \forall w \in \mathcal{W}$:
{
{\begin{equation}
\begin{aligned}
\label{eqn:DC_output_load_fault_semi} 
& \zeta P_{{\rm{G}}, w}(t) + P_{{\rm{E}}, w}(t) \\
= & P_{{\rm{L}},w}(t) + P_{{\rm{PR}}, w}(t) + \sum\nolimits_{z \in \Omega_{{\rm{pb}},w} }^{} S_{{\rm{P}},z}(t) P_{{\rm{vs}},z}(t)  \\
+ & \sum\nolimits_{z \in \Omega_{{\rm{sb}},w} } S_{{\rm{S}},z}(t) P_{{\rm{vs}},z}(t).
\end{aligned}
\end{equation}}
}

The relationship between the output power of the ESMs in different island parts can be described as, $\forall t \in \mathcal{T}$, $w \in \mathcal{W}$:
\begin{align}
\label{eqn:ESM_island_capacity} & E_{{\rm{e}},w}^{\min} \leqslant E_{{\rm{E}}, w}(t) \leqslant E_{{\rm{e}},w}^{\max}, \\
\label{eqn:ESM_island_range_1} &  P_{{\rm{e}},w}^{\min} \leqslant P_{{\rm{E}}, w}(t) \leqslant P_{{\rm{e}},w}^{\max},  \\
\label{eqn:ESM_island_relationship} &  E_{{\rm{E}}, w}(t) + P_{{\rm{E}}, w}(t) = E_{{\rm{E}}, w}(t+1),
\end{align}
where $Z_w^{}$ denotes the maximum number of ESMs that supply power to the $w$-th island part, which satisfies $\sum_{w \in \mathcal{W}}Z_w^{} = Z$. $E_{{\rm{e}},w}^{\min} = Z_w E_{\rm{e}}^{\min}, E_{{\rm{e}},w}^{\max} = Z_w E_{\rm{e}}^{\max}, P_{{\rm{e}},w}^{\min} = Z_w P_{\rm{e}}^{\min}$ and $ P_{{\rm{e}},w}^{\max} = Z_w P_{\rm{e}}^{\max}$.

Hence the control vector of each part in semi-island mode at time $t$ can be represented by
{{\begin{equation*}
\begin{aligned}
{\bm u}_w(t)  \stackrel{\vartriangle}{=} &(\bm \delta_{{\rm{g}},m}(t), \bm y_{{\rm{g}},m}(t), \bm{S}_{{\rm{P}}, x}(t), \bm{S}_{{\rm{S}}, y}(t), \bm P_{{\rm{g}},m}(t),  \\
& \bm P_{{\rm{e}}, n}(t), \bm P_{{\rm{pr}},r}(t), \rho_w (t) ), \; m \in \mathcal{M}_{w}, z\in \mathcal{Z}_{w}, \\
& n \in \mathcal{N}_{w}, r \in \mathcal{R}_{w}, x \in \Omega_{{\rm{pb}},w}, y \in \Omega_{{\rm{sb}},w}.
\end{aligned}
\end{equation*}}
}

\begin{table}[ht]  \scriptsize
\centering
\caption{{Differences of OPMS problem in three fault modes.}}
\label{tab:diff_three_modes}
\begin{tabular}{llll}
\hline
 & \begin{tabular}[c]{@{}l@{}}\textbf{Non-island}\\ \textbf{mode}\end{tabular} & \textbf{Island mode} & \begin{tabular}[c]{@{}l@{}}\textbf{Semi-island}\\ \textbf{mode}\end{tabular} \\ \hline
\begin{tabular}[c]{@{}l@{}}\textbf{Objective}\\ \textbf{function}\end{tabular} & same with \textbf{P1} & \begin{tabular}[c]{@{}l@{}}independent of each\\ other island part\end{tabular} & \begin{tabular}[c]{@{}l@{}}independent of each\\ other island part\end{tabular} \\ \hline
\begin{tabular}[c]{@{}l@{}}\textbf{Control} \\ \textbf{variables}\end{tabular} & \begin{tabular}[c]{@{}l@{}}some fixed\\ variables\end{tabular} & \begin{tabular}[c]{@{}l@{}}add $\rho_w$ in \\ each island part\end{tabular} & \begin{tabular}[c]{@{}l@{}}add $\rho_w$ and \\ redundant switches\end{tabular} \\ \hline
\begin{tabular}[c]{@{}l@{}}\textbf{No. of} \textbf{integer }\\ \textbf{variable}\end{tabular} & \tiny{$2MT$} & \begin{tabular}[c]{@{}l@{}} \tiny{$2MT$} \end{tabular} & \begin{tabular}[c]{@{}l@{}} \tiny{$(2M+2X)T$} \end{tabular} \\ \hline
\begin{tabular}[c]{@{}l@{}}\textbf{No. of} \\ \textbf{continuous }\\ \textbf{variable}\end{tabular} & \begin{tabular}[c]{@{}l@{}} \tiny{${(M+N}$} \\ \tiny{${ \ +2)T}$} \end{tabular} & \begin{tabular}[c]{@{}l@{}} \tiny{$(M+N+ $} \\ \tiny{$ \ W+2)T$} \end{tabular} & \begin{tabular}[c]{@{}l@{}} \tiny{$(M+N+$} \\ \tiny{$ \ W+2)T$} \end{tabular} \\ \hline
\begin{tabular}[c]{@{}l@{}}\scriptsize{\textbf{Power balance}}\\ \textbf{constraint}\end{tabular} & same with \textbf{P1} & \begin{tabular}[c]{@{}l@{}}independent of each \\ other island part\end{tabular} & \begin{tabular}[c]{@{}l@{}}\textbf{coupled} between\\ some island parts\end{tabular} \\ \hline
\end{tabular}
\end{table}

The OPMS problems for the three fault-modes are different with each other in objective function, control variables, and power balance constraints. The differences are summarized in Table \ref{tab:diff_three_modes}. 
There are $2X$ variables ($\bm S_{{\rm{P}},x}(t)$ and $\bm S_{{\rm{S}},y}(t)$) in semi-island mode more than that in other modes. There are also $W-1$ variables ($\bm {\rho}_w(t)$) more than that in normal mode. Constraint (\ref{eqn:DC_output_load_fault_semi}) is coupled between island parts in semi-island mode.
Thus, the OPMS problem for the semi-island mode is the most complex one.
Therefore, it is selected as the representative OPMSF problem for further analysis in the following subsection, which is formulated as:
\begin{equation*}
\begin{aligned}
\label{eqn:control_variable2}
\textbf{P2}:
\underset{ {\bm u}_w(t) }\min \; & \sum\nolimits_{w \in \mathcal{W}} \sum\nolimits_{t \in \mathcal{T}} C_{w}(t) \\
\text{s.t.} \quad &  (\ref{eqn:PG_limit})-(\ref{eqn:e_capacity_range}),
(\ref{eqn:total_load}), (\ref{eqn:redundant_switches}), (\ref{eqn:redundant_switches_time1}), (\ref{eqn:eeoi_limit}), (\ref{eqn:load_sheddding_penalty_xi}), (\ref{eqn:distance_ppr})-(\ref{eqn:ESM_island_relationship}).
\end{aligned}
\end{equation*}

\vspace{-3ex}
\subsection{Problem Transformation}
%

It is difficult to solve \textbf{P2} in four points: travel constrain (\ref{eqn:distance_ppr}) is non-convex; EEOI limitation is a fractional form; the feasibility is affected by faults; it is a MINLP mid-time scheduling problem. Firstly, the constraint (\ref{eqn:distance_ppr}) and EEOI limitation are transformed into convex forms.
Secondly, considering that the relaxed problem is infeasible, a further relaxation of travel constraint is developed to guarantee the feasibility. A sufficient condition is provided for that if \textbf{P2} is feasible, the two-step relaxed problem has the same optimal solution; if not, the maximum travel distance can be further obtained.

\subsubsection{Non-convex Form Transformation}
The OPMSF problem is non-convex due to constraint (\ref{eqn:eeoi_limit}) and (\ref{eqn:distance_ppr}) with $\beta = 3$.
$F{(P_{\rm{g}, m}(t))}$ is a convex form, and $V(t)$ is a concave form in $( \bm P_{{\rm{g}},m}(t), \bm P_{{\rm{e}}, n}(t), \bm P_{{\rm{pr}},r}(t) )$. 
Thus, an equivalent convex form of (\ref{eqn:eeoi_limit}) is obtained as:
\begin{equation}
\label{eqn:eeoi_limit_transform}
\begin{aligned}
{\sum\nolimits_{ g \in \mathcal{G} } F{(P_{\rm{g}, m}(t))}} - {\rm{EEOI}}^{\max} {F_{\rm{sl}} V(t)  \Delta t } \leqslant 0,
\end{aligned} 
\end{equation}

Then, based on convex relaxation, (\ref{eqn:distance_ppr}) is transformed into:
\begin{align}
\label{eqn:distance_ppr_relax}
& {D} - \sum\nolimits_{t \in \mathcal{T}} ( P_{\rm{PR}}{(t)} / \alpha )^{1/\beta} \Delta t \leqslant 0.
\end{align}

At last, \textbf{P3} is defined as \textbf{P2} with (\ref{eqn:eeoi_limit_transform}) and (\ref{eqn:distance_ppr_relax}) instead of  (\ref{eqn:eeoi_limit}) and (\ref{eqn:distance_ppr}). 
\begin{theorem}
\label{th:theorem 2}
The relaxed problem \textbf{P3} is exact, i.e., an optimal solution of \textbf{P3} is also optimal for the problem \textbf{P2}, if its optimal solutions satisfy (\ref{eqn:distance_ppr}).
\end{theorem}
\begin{proof}
 Please see Appendix \ref{append:theorem2}.
\end{proof}

\subsubsection{Feasibility-guaranteed Relaxation}
Considering that there may be no feasible solution of \textbf{P3} caused by faults, a further relaxed problem \textbf{P4} is formulated as:
\begin{align}
\textbf{P4}:
\underset{ {\bm u}_w(t), D_{\rm{d}} }\min \; & \sum\nolimits_{w \in \mathcal{W}} \sum\nolimits_{t \in \mathcal{T}} C_{w}(t) + h D_{\rm{d}} \nonumber \\
\label{eqn:relaxed_cons_P5} \text{s.t.} \quad & {D} - \sum\nolimits_{t \in \mathcal{T}} ( P_{\rm{PR}}{(t)} / \alpha )^{1/\beta} \Delta t \leqslant D_{\rm{d}}, \\
&  (\ref{eqn:PG_limit})-(\ref{eqn:e_capacity_range}),
(\ref{eqn:total_load}), (\ref{eqn:redundant_switches}), (\ref{eqn:redundant_switches_time1}), (\ref{eqn:load_sheddding_penalty_xi}), (\ref{eqn:range_ppr1})-(\ref{eqn:ESM_island_relationship}), (\ref{eqn:eeoi_limit_transform}) \nonumber
\end{align}
where $D_{\rm{d}}$ in (\ref{eqn:relaxed_cons_P5}) is the reduced travel distance that is a positive variable in \textbf{P4}, and $h$ denotes the penalty parameter of the reduced distance. To guarantee that $D_{\rm{d}} > 0$ only if \textbf{P3} is infeasible, a sufficient condition is derived as below.

\begin{proposition}
If \textbf{P3} has feasible solutions, the optimal solution of \textbf{P4} is also the optimal solution of \textbf{P3} when the penalty $h$ satisfies:
{{\begin{equation}
\label{eqn:parameter_condition}
h   > \dfrac {\beta \alpha^{1/\beta} \xi_{\rm{l}} } {\left( P_{\rm{PR}}^{\max} \right)^{1/\beta-1} }.
\end{equation}}}

Additionally, if \textbf{P3} has no feasible solution, $D-D_{\rm{d}}^{*}$ is the maximum travel distance that can be achieved in time $T$. $D_{\rm{d}}^{*}$ is optimal reduced distance that is obtained from \textbf{P4} with (\ref{eqn:parameter_condition}).
\end{proposition}
\begin{proof}
 Please see Appendix \ref{append:theorem3}.
\end{proof}

\begin{remark}
If there is not any port for repair and maintenance in the range of $D-D_{\rm{d}}^*$, the information $D-D_{\rm{d}}^*$ can be sent to the nearest port for rescue in advance.
\end{remark}

\vspace{-0ex}
\section{Power Management Algorithm Design}
\label{sec:algorithm_design}
An {optimal management algorithm} is designed based on BD \cite{floudas1995nonlinear,Sifuentes2007Hydrothermal,Network2016Nasri} to solve the \textbf{P4} over the entire time domain $\mathcal{T}$.
By the problem transformation in the subsection II-B, \textbf{P4} is a convex problem when the integer variables are determined. Thus, the {optimal management algorithm} based on BD can obtain the optimal solution since the sufficient condition for convergence to the global optimum is that the functions in the optimization problem satisfy some form of convexity conditions \cite{grossmann1990mixed}.
\textbf{P4} is decomposed into nonlinear programming (NLP) problem (subproblem) with fixed integer variables, and an integer linear programming (ILP) problem (master problem). 
The subproblem deals with the continuous variables and generates a set of dual variables to add Benders cuts in the master problem. After solving the master problem with Benders cuts, the optimal integer solution in this iteration is passed to subproblem. Hence, the subproblem and master problem are calculated iteratively to obtain the final optimal solution. 
Then, due to the slow convergence caused by the time-couping\cite{Sifuentes2007Hydrothermal}, 
an LNBD is proposed by decomposing the time-coupled constraints in the subproblem and adding the accelerating constraints in the master problem.

\vspace{-1ex}
\subsection{ Optimal Management Algorithm}
Based on the predicted data, the {optimal management algorithm} solves \textbf{P4} over the entire time domain $\mathcal{T}$. 
The subproblem and the master problem decomposed by BD are described as

\subsubsection{\textbf{Subproblem}}
the subproblem is defined as
\begin{subequations}
\label{eqn:subproblem_benders_p5}
\begin{align}
\textbf{P5}: 
\underset{ \bm u_{{\rm{sp}},w}^{(k)}(t)}
\min & C_{\rm{sp}}^{(k)} = \sum_{t \in \mathcal{T}} \sum_{w \in \mathcal{W}} \bigg( \sum_{m \in \mathcal{M}_w} \bigg( a_{{\rm{g}},m}  P_{{\rm{g}},m}^{(k)}(t)^2 \Delta t \bigg. \bigg. \nonumber \\ 
 & \bigg. \bigg. + b_{{\rm{g}},m}  P_{{\rm{g}},m}^{(k)}(t) \Delta t \bigg) + \xi_{\rm{e}} \sum_{n \in \mathcal{N}_w} C_{\rm{E}}(P_{{\rm{e}}, n}^{(k)}(t))  \nonumber \\
 & \bigg. \bigg. + \xi_{\rm{l}} \sum_{w \in \mathcal{W}} C_{\rm{L}}(\rho_w^{(k)}(t)) \bigg) + hD_{\rm{d}}^{(k)} \nonumber \\
\text{s.t.} \quad 
\label{eqn:delta_bd} & \bm \delta_{{\rm{g}},m}(t) = \bm \delta_{{\rm{g}},m}^{(k-1)}(t) : \bm \lambda_{\delta,m}(t), \\
\label{eqn:sp_bd} & \bm{S}_{{\rm{P}}, x}(t) = \bm{S}_{{\rm{P}}, x}^{(k-1)}(t) : \bm \lambda_{{\rm{P}},x}(t), \\ 
\label{eqn:ss_bd} & \bm{S}_{{\rm{S}}, y}(t) = \bm{S}_{{\rm{S}}, y}^{(k-1)}(t) : \bm \lambda_{{\rm{S}},y}(t), \\
&  (\ref{eqn:PG_limit}), (\ref{eqn:P_Rate_limit}), (\ref{eqn:pe_range}), (\ref{eqn:e_capacity_range}), (\ref{eqn:total_load}), (\ref{eqn:load_sheddding_penalty_xi}), (\ref{eqn:range_ppr1})-(\ref{eqn:eeoi_limit_transform}), 
(\ref{eqn:relaxed_cons_P5}), (\ref{eqn:parameter_condition}). \nonumber
\end{align}
\end{subequations}

The control vector of the subproblem $\bm u_{{\rm{sp}},w}^{(k)}(t)$ in each part at time $t$ can be represented by
\begin{equation*}
\begin{aligned}
\label{eqn:variable_in_sumproblem}
\bm u_{{\rm{sp}},w}^{(k)}(t) \stackrel{\vartriangle}{=} & (\bm P_{{\rm{g}},m}^{(k)}(t), \bm P_{{\rm{pr}},r}^{(k)}(t), \bm P_{{\rm{e}}, n}^{(k)}(t), \rho_w^{(k)}(t), D_{\rm{d}}^{(k)} ),  \big. \\
& \big. m \in \mathcal{M}_{w}, n \in \mathcal{N}_{w}, r \in \mathcal{R}_{w}, \forall w, \forall t.
\end{aligned}
\end{equation*}

$\bm \delta_{{\rm{g}},m}^{(k-1)}(t)$,  $\bm{S}^{(k-1)}_{{\rm{P}}, x}(t)$, $\bm{S}^{(k-1)}_{{\rm{S}}, y}(t)$, and $ \bm P_{{\rm{g}},m}^{(k-1)}(t)$ are fixed as the input data to the subproblem which is computed in the master problem.
$\bm \delta_{{\rm{g}},m}(t)$ controls the constraints (\ref{eqn:PG_limit}) and (\ref{eqn:start-up_detect}).
$\bm{S}_{{\rm{P}}, x}(t)$ and $\bm{S}_{{\rm{S}}, y}(t)$ affect the constraint (\ref{eqn:redundant_switches}), (\ref{eqn:redundant_switches_time1}), and (\ref{eqn:DC_output_load_fault_semi}).
$\bm \lambda_{\delta,m}(t)$, 
$\bm {\lambda}_{{\rm{P}},x}(t)$, $\bm {\lambda}_{{\rm{S}},y}(t) $, and 
$\bm \lambda_{{\rm{g}},m}(t)$ are defined as $\bm \theta(t)$, which is a set of dual variables of $ \bm u_{{\rm{mp}},w}(t) $. $\bm \theta(t)$ provides sensitivties to be used in constructing Benders' cut $\mu^{(k)}$ for the master problem.
The upper bound for the optimal objective value of \textbf{P4} at iteration $k$ is calculated by
\begin{equation}
\label{eqn:benders_up}
\begin{aligned}
\overline C^{(k)} & = 
\sum_{t \in \mathcal{T}} 
\sum_{m \in \mathcal{M}} c_{{\rm{g}},m} \bigg( \delta_{{\rm{g}},m}^{(k-1)}(t) \Delta t \bigg)
 + C_{\rm{sp}}^{(k)}.
 \end{aligned}
\end{equation}

\subsubsection{\textbf{Master problem}}
The master problem includes the minimization of the third term of the fuel consumption cost of generators and benders‘ cut. Two constraints are added to improve the convergence, $\forall w \in \mathcal{W}, \forall t \in \mathcal{T} $:
\begin{align}
\label{eqn:accele_max}& P_{{\rm{G}},w}^{\max} + P_{{\rm{E}},w}^{\max} \geqslant P_{{\rm{PR}},w}^{(v)}(t) + P_{{\rm{VS}},w}(t) , \\
\label{eqn:accele_min}& P_{{\rm{G}},w}^{\min} + P_{{\rm{E}},w}^{\min} \leqslant 
P_{{\rm{VS}},w}(t)+ P_{{\rm{NO}},w}(t) , \\
\label{eqn:lower_bdm} & \mu^{(k)}  \geqslant \underline{\mu}^{}, 
\end{align}
where the term in (\ref{eqn:accele_max}) is defined as $P_{{\rm{G}},w}^{\max} = \sum_{m \in \mathcal{M}_{w}} \delta_{{\rm{g}},m}(t)P_{{\rm{g}},m}^{\max} $ and the term in (\ref{eqn:accele_min}) $ P_{{\rm{G}},w}^{\min} = \sum_{m \in \mathcal{M}_{w}} \delta_{{\rm{g}},m}(t)P_{{\rm{g}},m}^{\min} $. Constraints (\ref{eqn:accele_max}) and (\ref{eqn:accele_min}) guarantee that the range of the generation power covers the power of load demand at each time $t$. Consequently, a part of invalid solutions can be eliminated. 
Constraint (\ref{eqn:lower_bdm}) gives a lower bound on $\mu^{(k)}$ to avoid the search for the invalid solution below that bound.
Hence, the master problem can be described as:
\begin{subequations}
\begin{align}
\textbf{P6}: \underset{\bm u_{{\rm{mp}},w}^{(k)}(t)} \min &  { \underline C^{(k)} = 
\sum_{t \in \mathcal{T}} 
\sum_{w \in \mathcal{W}}  \sum_{m \in \mathcal{M}_w} \bigg( c_{{\rm{g}},m} \delta_{{\rm{g}},m}^{(k)}(t) \Delta t \bigg)
+  \mu^{(k)}} \nonumber \\
\text{s.t.} \quad & \mu^{(k)} \geqslant  \sum_{m \in \mathcal{M} } \sum_{t \in \mathcal{T} } \lambda_{\delta,m}(t) \left( \delta_{{\rm{g}},m}^{(k)}(t) -  \delta_{{\rm{g}},m}^{(v)}(t) \right) \nonumber \\
& \quad \quad  + \sum_{x \in \Omega_{{\rm{pb}},x} } \sum_{t \in \mathcal{T} } \lambda_{{\rm{P}},x}(t) \left( {S}_{{\rm{P}}, x}^{(k)}(t) - {S}_{{\rm{P}}, x}^{(v)}(t) \right) \nonumber \\
& \quad \quad  + \sum_{y \in \Omega_{{\rm{sb}},y} } \sum_{t \in \mathcal{T} } \lambda_{{\rm{S}},y}(t) \left( {S}_{{\rm{S}}, y}^{(k)}(t) - {S}_{{\rm{S}}, y}^{(v)}(t) \right)  \nonumber \\
\label{eqn:cut_bdm} & \quad \quad + C_{\rm{sp}}^{(v)}; v \in [ 1, \cdots, k-1 ],\\
& (\ref{eqn:start-up_detect}), (\ref{eqn:min_time_operation}), (\ref{eqn:redundant_switches}), (\ref{eqn:redundant_switches_time1}), (\ref{eqn:accele_max})-(\ref{eqn:lower_bdm}). \nonumber 
\end{align}
\end{subequations}

$\underline C^{(k)}$ is the lower bound for the objective value of \textbf{P4}. In each iteration, $\underline C^{(k)}$ is improved by the Benders' cut $\mu^{(k)}$. The control vector of the master problem $\bm u_{{\rm{mp}},w}(t)$ in each island part at time $t$ can be represented as
\begin{equation*}
\begin{aligned}
\label{eqn:variable_in_master_problem}
\bm u_{{\rm{mp}},w}^{(k)}(t) \stackrel{\vartriangle}{=} & ( \bm \delta_{{\rm{g}},m}^{(k)}(t), \bm y_{{\rm{g}},m}^{(k)}(t),
\bm{S}_{{\rm{P}}, x}^{(k)}(t), \bm{S}_{{\rm{S}}, y}^{(k)}(t), \\
& \bm y_{s,x}^{(k)}(t) ), m \in \mathcal{M}_{w},   x \in \Omega_{{\rm{pb}},w}, y \in \Omega_{{\rm{sb}},w}.
\end{aligned}
\end{equation*}

The {optimal management algorithm} based on BD converges to the optimal solution by iteratively calculating \textbf{P6} and \textbf{P5}. {The iteration flowchart is shown in Fig. \ref{fig:algorithm_structure}. }

\begin{figure}[ht]
\begin{center}
\vspace{-2ex}
\includegraphics[width= 0.99 \linewidth]{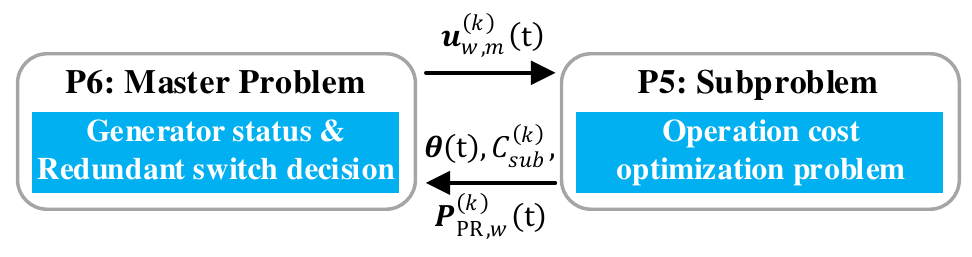} 
\caption{ {Flowchart of optimal management algorithm.}}
\captionsetup{justification=centering}
\label{fig:algorithm_structure}
\vspace{-3ex}
\end{center}  
\end{figure}


\subsection{Low-complexity Near-optimal Algorithm}
Since the {optimal management algorithm} is exponential in $T$, an LNBD is developed to deal with the slow convergence.

\subsubsection{\textbf{Subproblem}} 
Here it is needed to decompose the ESM and travel constraints in (\ref{eqn:ESM_island_relationship}) and (\ref{eqn:distance_ppr_relax}).
Thus, the power of ESMs are allocated by adjusting the upper bound based on the estimated average load demand $\bar P_{{\rm{L}},w}(t)$ and average power of ESMs. 
If $\sum_{m \in \mathcal{M}_{w}} \delta{{\rm{g}},m}(t) >0$, (\ref{eqn:ESM_island_relationship}) is replaced by
\begin{equation}
\begin{aligned}
\label{eqn:P_E_online}
P_{{\rm{E}},w}(t) \leqslant & \bar E_{w} + \varphi(t) \Delta P_{{\rm{L}},w}(t), \text{if} \ \sum\nolimits_{m \in \mathcal{M}_{w}} \delta{{\rm{g}},m}(t) >0; \\
\end{aligned}
\end{equation}
where $ \bar E_{w}^{} = ( E_{w} - E_{{\rm{e}},w}^{\min})/T$,
otherwise, there is no other additional constraint. $\Delta P_{{\rm{L}},w}(t) = P_{{\rm{L}},w}(t) - \bar P_{{\rm{L}},w}$, and $ \varphi(t) \in [0,1] $ denotes the parameter for adjusting the power of ESMs. Besides, based on the $\bar P_{{\rm{L}},w}(t)$, (\ref{eqn:distance_ppr_relax}) can be decomposed into
\begin{equation}
\begin{aligned}
\label{eqn:online_distance}
& D(t) + D_{\rm{d}}(t) - \left( {P_{\rm{PR}}(t)} /{\alpha} \right)^{{1}/{\beta}} \Delta t \leqslant 0, \\
\end{aligned}
\end{equation}
where $D(t)$ is calculated by $ D(t) = \left( ({ \bar P_{\rm{PR}} - (1 - \varphi(t)) \Delta P_{{\rm{L}},w}(t) }) /{\alpha} \right)^{\frac{1}{\beta}} \Delta t $.  
It has to satisfy $\varphi (T) = 1$ to complete the voyage.
When $D_{\rm{d}}(t)>0$ and $D_{\rm{d}}(t)=0$ all exist, $D_{\rm{d}}(t)>0$ shows that it cannot achieve $D(t)$, and $D_{\rm{d}}(t)=0$ means the power at this time is enough to finish $D(t)$. Thus, in this case $D(t)$ is updated as $D(t) + \sum_{t\in \mathcal{T} }D_{\rm{d}}(t)/T$.
Thus, the subproblem is redefined as
\begin{equation*}
\begin{aligned}
\textbf{P7}:
\underset{ \bm u_{{\rm{sp}},w}^{(k)}(t)} \min 
& C_{\rm{sp}}^{(k)}  = \sum_{w \in \mathcal{W}}  \bigg( \sum_{m \in \mathcal{M}_w} a_{{\rm{g}},m} P_{{\rm{g}},m}^{(k)}(t)^2 \Delta t  \bigg. \\ 
 & \quad \quad \bigg. \bigg. + b_{{\rm{g}},m}  P_{{\rm{g}},m}^{(k)}(t) \Delta t + \xi_{\rm{e}} \sum_{n \in \mathcal{N}_w} C_{\rm{E}}(P_{{\rm{e}}, n}^{(k)}(t))  \nonumber \\
 & \quad \quad \bigg. \bigg. + \xi_{\rm{l}}  C_{\rm{L}}(\rho_w^{(k)}(t)) \bigg) + hD_{\rm{d}}(t)^{(k)} \nonumber \\
\text{s.t.} \quad &  (\ref{eqn:PG_limit}), (\ref{eqn:P_Rate_limit}), (\ref{eqn:pe_range}), (\ref{eqn:e_capacity_range}), (\ref{eqn:total_load}), (\ref{eqn:load_sheddding_penalty_xi}), (\ref{eqn:range_ppr1})-(\ref{eqn:ESM_island_range_1}), (\ref{eqn:eeoi_limit_transform}),  \\
&  (\ref{eqn:relaxed_cons_P5})-(31), (\ref{eqn:P_E_online}),  (\ref{eqn:online_distance}).
\end{aligned}
\end{equation*}

\subsubsection{\textbf{Master problem}}
Thus, the master problem is redefined as:
\begin{equation*}
\begin{aligned}
\label{eqn:masterproblem_benders_online}
\textbf{P8}:
\underset{ \bm u_{{\rm{mp}},w}(t)} \min \; & { \underline C^{(k)}(t) = 
\sum_{t \in \mathcal{T}} \sum_{m \in \mathcal{M}} \bigg( c_{{\rm{g}},m} \delta_{{\rm{g}},m}^{(k-1)}(t) \Delta t \bigg) 
+ \mu^{(k)} } \\
\text{s.t.} \quad
& (\ref{eqn:start-up_detect}), (\ref{eqn:min_time_operation}), (\ref{eqn:redundant_switches}), (\ref{eqn:redundant_switches_time1}), (\ref{eqn:cut_bdm}), (\ref{eqn:accele_max})-(\ref{eqn:lower_bdm}).
\end{aligned}
\end{equation*}

The procedure of the LNBD is given in Algorithm \ref{alg:offline}.

%
\begin{algorithm}[htb]  
  \caption{LNBD Algorithm for OPMSF problem}
  \label{alg:offline}
  \LinesNumbered
  \KwIn{    $ \overline C = -\infty $, $ \underline C = + \infty $, and $\epsilon = 10^{-2}$\;}
  \KwOut{ $ \bm u_{{\rm{sp}},w}(t) $, $ \bm u_{{\rm{mp}},w}(t) $, $\rho_{w}$, and $D_{\rm{d}}$\;}
  Set $k \leftarrow 1$ \;   
  \Repeat {$ \overline C-\underline C < \epsilon $}
  {  
    Obtain $\bm u_{{\rm{mp}}, w}^{(k)}(t)$ and $\underline C^{(k)} $ by solving \textbf{P8}\;
    \If {$\underline C^{(k)} > \underline C $}
    {
      Set $\underline C = \underline C^{(k)} $\;
    } 
    \Repeat {each $D_{\rm{d}}(t) = 0$ or each $D_{\rm{d}}(t) > 0$}
    {  
      Update $D_{\rm{d}}(t)$ and solve $T$ \textbf{P7} problems in sequence based on $\bm u_{{\rm{mp}},w}^{(k)}(t) $\;
    } 
    Obtain $\bm u_{{\rm{sp}},w}^{(k)}(t) $ and $\bm \theta^{(k)}$ by solving \textbf{P7}\;
    Generate cut (\ref{eqn:cut_bdm}) based on $\bm u_{{\rm{sp}},w}^{(k)}(t) $ and $\bm \theta^{(k)}$, and add it into \textbf{P8}\;
    Calculate the upper bound $\overline C^{(k)}$ of \textbf{P4} by (\ref{eqn:benders_up})\;
    Set $k \leftarrow k+1 $\;
    \If {$\overline C^{(k)} < \overline C $}
    {  
      Set $\overline C = \overline C^{(k)} $\;
    }
  }
\end{algorithm}

%
%

\subsection{Complexity Analysis}
\label{sec:complexity_analysis}
Since the subproblems and the master problems are convex NLP problems and ILP problems, respectively, they are much easier to solve than the MINLP problem \textbf{P4}.
To compare the performance of two algorithms, computational complexity needs to be addressed.
Starting from the {optimal management algorithm} based on BD, the subproblem can be solved in polynomial time \cite{nemirovski2004interior}, and the computational complexity of subproblem is $\mathcal{O}( T_{\rm{sp}})$, where $T_{\rm{sp}}$ denotes the number of iterations required in the subproblem \textbf{P5} and $\mathcal{O}(\cdot)$ is the big-$\mathcal{O}$ notation.
In MVDC SPS, the startup time typically ranges from one to five minutes \cite{ieeestd4512017}. The minimum operation time $T_{{{m}}}^{\min}$ is determined by the startup time that brings generator online. $T_{{{m}}}^{\min}$ is larger than the time interval $\Delta t$ that is 0.5 or 1 hour in \cite{kanellos2014optimal,Kanellos2016Smart,Kanellos2017cost,Kanellos2014Optimaldemand,shang2016economic}. 
Thus, $T_{{{m}}}^{\min}$ is set to $\Delta t$. Based on that, the ILP master problem has a non-polynomial complexity, and its computational complexity is $\mathcal{O}(2^{(M+X)T})$.
Considering the iteration in BD, the overall algorithm complexity is $\mathcal{O}( (T_{\rm{sp}} + 2^{(M+X)T} ) K )$, where $K$ represents the number of iterations required for BD to converge. 
$K$ is related to the strength of the Bender's cuts, and has a positive correlation with the complexity of master problem \cite{rahmaniani2017benders}.
It is observed that the {optimal management algorithm} based BD is exponential in $M$, $X$, and $T$. 
Due to the time decomposing in the LNBD, the overall algorithm complexity is reduced to $\mathcal{O}( ( \bar T_{\rm{sp}} + 2^{M+X} ) T \bar K )$, where $\bar K$ denotes the number of iterations required for LNBD to converge.
Thus, it is exponential in $M$ and $X$, and is polynomial in $T$. $\bar T_{\rm{sp}}$ denotes the number of iterations required in the subproblem \textbf{P7}.
Due to the limited generators and redundant switches in SPS, $M$ and $X$ are small, and then the computational complexity is mostly related to $T$. Hence, LNBD can be considered to have a polynomial time computational complexity.
\section{Simulations}
\label{sec:simulation}
In this section, an MVDC SPS and the simulation setup are described in detail. Then the proposed solutions are tested by applying it to the MVDC SPS.

\vspace{-2ex}
\subsection{Simulation Setup}
The subproblems are solved using the SDPT3 from the CVX package \cite{grant2008cvx}, which operates on an Intel CORE i5 3.4 GHz machine with 8 GB RAM. The master problem is solved by branch and cut.

\begin{figure}[ht]
\vspace{-2ex}
\begin{center}
\setlength{\belowcaptionskip}{-0.3cm}
\includegraphics[width= 0.47 \textwidth]{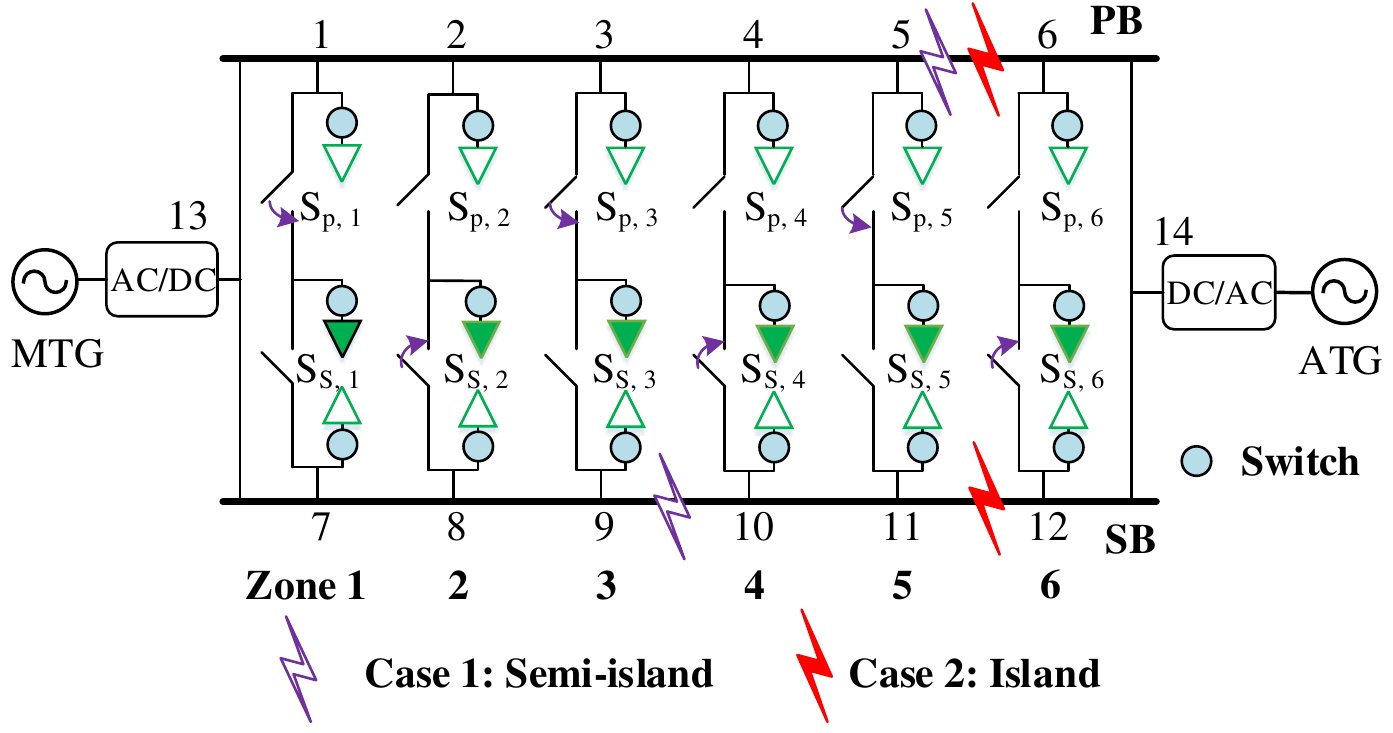} 
\caption{ MVDC SPS model and the fault scenario. } 
\captionsetup{justification=centering} 
\label{fig:fault_scenario}
\end{center}
\vspace{-2ex}
\end{figure}

\begin{table}[ht] \scriptsize
\centering
\vspace{-2ex}
\caption{Simulation parameters.}
\label{tab:Simulation_parameters}
\begin{tabular}{llll}
\hline \hline
\multicolumn{1}{c}{}   &  MTGs  & ATGs & ESMs \\ 
\hline
Normal power (MW)            & 8   &  4  & 0.5    \\
Technical maximum (p.u.)     & 1        & 1     & 1       \\
Technical minimum (p.u.)     & 0.15      & 0.1   & {0.1}       \\
Ramp rate limit (p.u./$\Delta t$) & $ \pm 90\% $ & $ \pm 90\% $ & $ \pm 100\% $ \\
\multirow{4}{*}{Cost function parameters} 
& {$a_{{\rm{g}},1}=13.5$}  
& $a_{{\rm{g}},2}=6$   
& $a_{\rm{lc}}=1$  \\
& $b_{{\rm{g}},1}=10$    
& $b_{{\rm{g}},2}=30$  & $c_{\rm{lc}}=0.5$ \\
& $c_{{\rm{g}},1}=300$   
& $c_{{\rm{g}},1}=250$ 
&      \\
& $T_{ 1}^{\min}=1$ 
& $T_{ 2}^{\min}=1$    
& $\xi_{\rm{e}} = 1$ \\ \hline
\multicolumn{1}{c}{}     & \multicolumn{2}{c}{Other Modules} 
& \multicolumn{1}{c}{Others}   \\ \hline
Speed \& power parameter & \multicolumn{1}{l}{$\alpha = 2.2e^{-3}$} 
& \multicolumn{1}{l}{$\beta = 3$}  
& \multicolumn{1}{l}{ $ T_{s}^{\min} = 1 $ } \\ 
Speed (kn)  
& \multicolumn{1}{l}{$V^{\min} = 0$} 
& \multicolumn{1}{l}{$V^{\max} = 17$} 
& \multicolumn{1}{l}{ $\Delta t$ = 1hour } \\ 
{EEOI}
& \multicolumn{2}{l}{${\rm{EEOI}}^{\max} = 23$ ${\rm{gCO_2}}/{\rm{tnkn}}$} 
& \multicolumn{1}{l}{ $F_{\rm{sl}}$ = 30} \\ 
{Penalty parameters}
&  $\xi_l = 265$ &  $h = 1.15e+3$
& \\ 
\hline \hline
\end{tabular}
\vspace{-2ex}
\end{table}



{\subsection{Case Study 1: Semi-island Mode}}
A MVDC SPS with $ Z = 6 $ and $T = 10 $ is shown in Fig. \ref{fig:fault_scenario}. 
There are one MTG in zone 1, one ATG in zone 6, and four ESMs in zone 1, 3, 4 and 6. 
The propulsion module is located in zone 2. 
The detail parameters are shown in Table \ref{tab:Simulation_parameters} that refers to \cite{kanellos2014optimal}.
The scenario is a 10-h voyage with $D = 120$ nm. 
A fault scenario that belongs to semi-island mode occurs at the begin as shown in Fig. \ref{fig:fault_scenario}. 
The initial switch configuration is given in Fig. \ref{fig:fault_scenario}.
The forecast service load is plotted with the solid brown line in Fig. \ref{fig:reorganized_plan_fault}.
The scale factors of vital, semi-vital and non-vital loads in the total service loads are set as $0.3$, $0.5$, $0.2$. The non-vital loads $P_{\rm{no}}(t)$ can be considered as the upper bound of load shedding $P_{\rm{ls}}(t)$ due to the limit in (\ref{eqn:total_load}), which is plotted with the black short dotted line. 

\begin{figure}[htbp]
\centering
\vspace{-1ex}
\includegraphics[width=0.485 \textwidth]{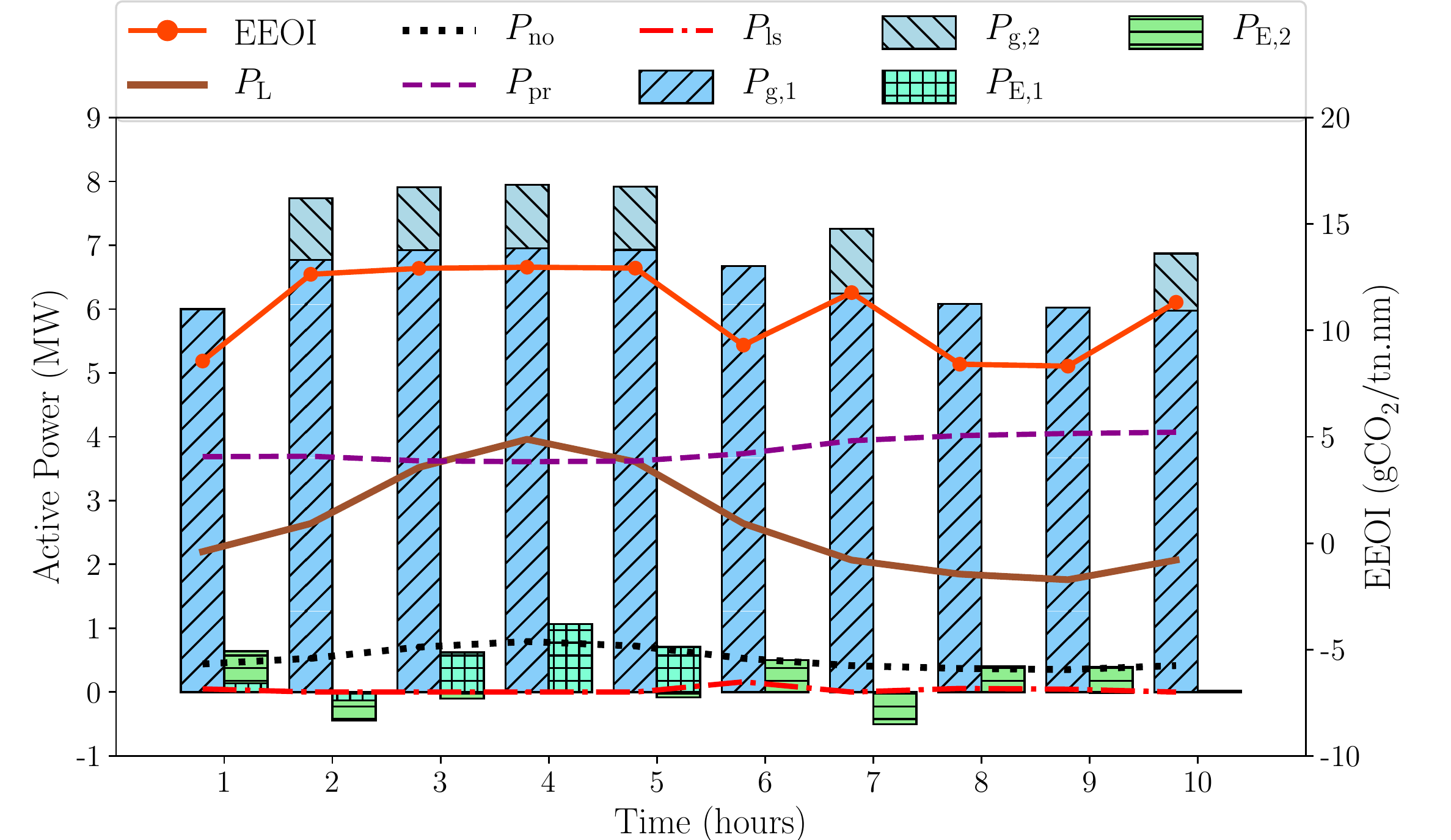}
\caption{  Power schedule with two-side management in semi-island mode. }
\label{fig:reorganized_plan_fault}
\vspace{-2ex}
\end{figure}

\subsubsection{{Performance Analysis of Two-side Management}}
\label{sec:simulation_operation}
The schedule with two-side adjustment methods obtained by the {optimal management algorithm} is shown in Fig. \ref{fig:reorganized_plan_fault}.
The propulsion power is changed at different time, which shows that PPA already works.
{Load shedding works mainly at the 6-th time interval.} The load-shedding amount $P_{\rm{ls}}(t)$ is less than the non-vital load demand $P_{\rm{no}}(t)$.
Thus, it can be observed that PPA and load shedding already work for reducing operating cost and guaranteeing the system safety. {The ESMs in island 2 absorbed power when ATG works at the 2-nd and 7-th time intervals.} {At the 1-st, 6-th, 8-th, and 9-th time intervals, the ESMs in island 2 produced power for loads to avoid mechanical damage of ATG caused by working below $P_{{\rm{g}},2}^{\min}$. $P_{{\rm E}, w}(t)$ is plotted instead of $P_{{\rm e}, n}(t)$ in Fig. \ref{fig:reorganized_plan_fault} to Fig. \ref{fig:power_scheduling_island}. The reason has two aspects. First, $P_{{\rm e}, n}(t)$ does not have a good visibility due to the smaller value. Second, $P_{{\rm e}, n}(t)$ in a island part has a same output power due to the same initial parameters and the quadratic cost function $C_{\rm E}(P_{{\rm e}, n}(t)$.}
%


{The key difference is the coordination mechanism and the feasibility-guaranteed mechanism to reduce the fault effects. Thus, they are verified by the OPMSF problem with different travel distance $D$ compared with the method in [10]. The method in [10] is a two-side management method including GS, ESMC, and PPA.} The results are shown in Table \ref{tab:different_distance}. The total load-shedding amount is defined as $ P_{\rm{LS}} = \sum_{t \in \mathcal{T}} \sum_{w \in \mathcal{W} } \rho_w(t)P_{{\rm{NO}},w}(t)$. 
The number of reconfiguring switches is denoted as $N_{\rm{rs}}$. 
It can be observed that load shedding works and reconfiguration always works in the different distance. Then, based on the results of last three columns, it can be obtained that the maximum travel distance $142.8$nm and the maximum amount of non-vital loads to reduce the fault effects is $5.38$ MW. 
Since it is without load shedding, reconfiguration, and feasibility-guaranteed relaxation, the algorithm in \cite{kanellos2014optimal} cannot solve the problem.
{ The load shedding also works at $D=100$nm, because the ESMC and GS cannot guarantee the safety operation in the second island part where the ATG locates.
Thus, the coordination mechanism and the feasibility-guaranteed mechanism have good effects for the OPMSF problem in semi-island mode.} 

\begin{table}[ht]
\vspace{-4ex}
\centering
\caption{{Two-side management with different distances.}}
\label{tab:different_distance}
\begin{threeparttable}
\begin{tabular}{llllllll}
\hline \hline
\multirow{3}{*}{} & & \multicolumn{6}{c}{Distance $D$ (nm)} \\ \cline{3-7} 
                  & & 100 & 120  &  140 &  160 &  180  \\ \hline
\multirow{1}{*}{{\begin{tabular}[c]{@{}c@{}}Method in \cite{kanellos2014optimal}\end{tabular}} } 
& Cost \tnote{1} &  -\tnote{2} & -  & - & - &  - \\ 
\hline 
\multirow{5}{*}{{\begin{tabular}[c]{@{}c@{}} Optimal manage- \\ ment algorithm \end{tabular}}} & Cost \tnote{1} & 1.09 &  1.29  &  1.66 &  1.66 &  1.66 \\ 
& $P_{\rm{LS}}$ (MW) & 0.31 & 0.31 & 1.70 & 5.38 & 5.38 \\ 
& $N_{\rm{rs}}$ &  5 &  5 &  1  &  1  & 1 \\ 
& $D_{\rm{d}}$ (nm) &  0 &  0 &  0    &  17.2   & 37.2 \\ 
\hline \hline
\end{tabular}
        \footnotesize
        \tnote{1} Cost unit: $10^4$ m.u. \;
        \tnote{2} - represents no feasible solution.
\end{threeparttable}
\vspace{-3ex}
\end{table}

\subsubsection{Performance Analysis of Algorithms}
\label{sec:simulation_LNBD}
To illustrate the performance of the proposed algorithms when the feasibility-guaranteed relaxation works, the test is conducted under $D=160$, different adjustment methods, and different parameter.

\paragraph{Performance under different $\varphi(t)$ }
Firstly, the optimality under different $\varphi(t)$ is tested in failure mode. The results are shown in Table \ref{tab:operating_cost_comparison_with_varphi}. $t_c$ denotes the total computation time. {It can be observed that they has a similar performance when $\varphi(t)$ is at the range of $0.3 \sim 0.8$. $t_c$ can be reduced from 198s to 85.2s.} To adopt ESMC and PPA together, $\varphi(t)$ is set to $0.5$.

\begin{table}[ht] 
\centering
\vspace{-4ex}
\caption{Performance comparison under different $\varphi(t)$.}
\label{tab:operating_cost_comparison_with_varphi}
\begin{threeparttable}
\begin{tabular}{cccccccccc}
\hline \hline
& \multirow{3}{*}{\begin{tabular}[c]{@{}c@{}} \scriptsize Optimal \\ mgmt. \\  algorithm\end{tabular}} & \multicolumn{8}{c}{\multirow{2}{*}{ LNBD with different $\varphi(t)$} }   \\ 
&  & &  &  &  &  &  &  \\ \cline{3-10}
&  
& 0.1 
& 0.2 & 0.3 & 0.5 & 0.6 & 0.7 & 0.8 \\ \hline
Cost \tnote{1}   
&  1.66 & 1.64 & 1.63 & 1.63  & 1.63 & 1.63 & 1.63 & 1.63 \\
$P_{\rm{LS}}$ \tnote{2}  
&  5.38 & 5.38 & 5.38 & 5.38 & 5.38 & 5.38 & 5.38 & 5.38 \\ 
$D_{\rm{d}}$ \tnote{3}  
&  17.2 & 21.2 & 21.1 & 21.0 & 20.9 & 20.9 & 20.9 & 20.8 \\
\multicolumn{1}{c}{\begin{tabular}[c]{@{}l@{}}$t_c$ \tnote{4} \end{tabular}} 
& 198 & 88.6 & 87.5 & 86.0 & 87.1 & 85.2 & 85.6 & 87.3 \\ 
\hline \hline
\end{tabular}
        \footnotesize
        \tnote{1} Cost unit: $10^3$ m.u;
        \tnote{2} Unit: MW;
        \tnote{3} Unit: nm; 
        \tnote{4} Unit: s. 
\end{threeparttable}
\vspace{-2ex}
\end{table}

The main difference of LNBD is that it adopts suboptimal power allocation in (\ref{eqn:P_E_online}) and (\ref{eqn:online_distance}). It can be known that there are little difference of $P_{{\rm{E}},1}(t)$ and $P_{\rm{PR}}(t)$ from the Fig. \ref{fig:pe_ppr_comparison}. $P_{{\rm{E}},2}$ of two algorithms has a large difference since the low load demand affects the generator state. 
\begin{figure}[htbp]
\centering
\vspace{-2ex}
\includegraphics[width=0.48 \textwidth]{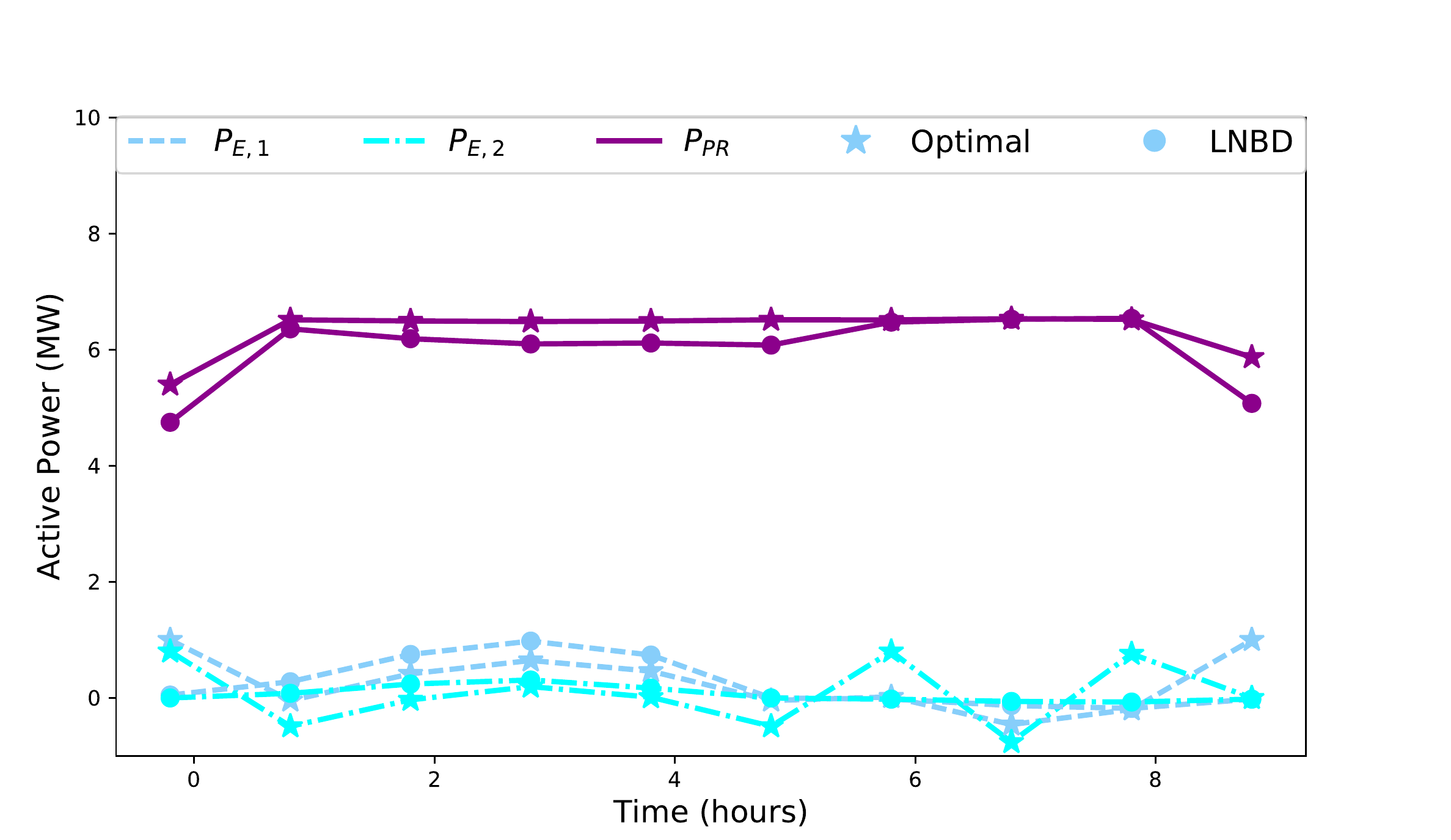}
\caption{ Comparison of $P_{{\rm{E}},w}(t)$ and $P_{\rm{PR}}(t)$ by two algorithms ($\varphi = 0.5$). }
\label{fig:pe_ppr_comparison}
\vspace{-2ex}
\end{figure}

\begin{table}[h] \scriptsize  
\centering
\vspace{-3ex}
\caption{ Performance comparison between different algorithms.}
\label{tab:operating_cost_comparison}
\begin{threeparttable}
\begin{tabular}{cc|ccc|cc}
\hline \hline
\multicolumn{2}{c|}{Status} 
& \multicolumn{3}{c|}{Normal} & \multicolumn{2}{c}{\begin{tabular}[c]{@{}c@{}} Failure  \end{tabular}}\\ \cline{1-7}
\multicolumn{2}{c|}{\begin{tabular}[c]{@{}c@{}} Adjustment \\ methods \end{tabular} } & \begin{tabular}[c]{@{}c@{}}  Only \\ GS \end{tabular} & \begin{tabular}[c]{@{}c@{}} w/ ESMC, \\ GS \& PPA \end{tabular}
& \begin{tabular}[c]{@{}c@{}} Two-side \\ mgmt. \end{tabular}  
& \begin{tabular}[c]{@{}c@{}} w/ ESMC, \\ GS \& PPA  \end{tabular}
&  \begin{tabular}[c]{@{}c@{}} Two-side \\ mgmt. \end{tabular} \\ \hline
\multicolumn{2}{c|}{\begin{tabular}[c]{@{}c@{}} No. Var. (Con- \\ tinuous/Integer) \end{tabular}} 
& \begin{tabular}[c]{@{}c@{}} 60 \\ (20/40) \end{tabular} 
& \begin{tabular}[c]{@{}c@{}} 111 \\ (71/40) \end{tabular}
& \begin{tabular}[c]{@{}c@{}} 121 \\ (81/40) \end{tabular} 
& \begin{tabular}[c]{@{}c@{}} 111 \\ (71/40) \end{tabular}
&  \begin{tabular}[c]{@{}c@{}} 191 \\ (91/100) \end{tabular} \\ \hline
\multirow{6}{*}{\rotatebox{90}{\begin{tabular}[c]{@{}c@{}} Optimal manage- \\ ment algorithm \end{tabular}} } 
& \begin{tabular}[c]{@{}l@{}} Cost \tnote{1}\end{tabular} 
& -  & 1.84  & 1.84  & 1.63  & 1.65   \\ 
& \begin{tabular}[c]{@{}l@{}} Diff. \tnote{2} \end{tabular}
& -  & \bf{0}   & 0   &  -11.4  & -10.3   \\ \cline{2-7}
& \begin{tabular}[c]{@{}l@{}} {\tiny{$P_{\rm{LS}}$}} \tnote{3} \end{tabular} 
& -  & 0  & 7.17  & 0   & 5.38  \\ \cline{2-7}
& \begin{tabular}[c]{@{}l@{}} {\tiny{$D_{\rm{d}}$}} \tnote{4} \end{tabular} 
& -  & 5.43  & 0.98 & 25.3   & 17.2  \\ \cline{2-7}
& \begin{tabular}[c]{@{}l@{}} $ t_c $ \tnote{5} \end{tabular}
& -   & 36.9   & 67.3  & 74.9   &  198 \\
& \begin{tabular}[c]{@{}l@{}} Iter. \end{tabular}
& -   & 18   &  25  & 27   &  109 \\ \hline
\multirow{6}{*}{\rotatebox{90}{\begin{tabular}[c]{@{}c@{}} LNBD \end{tabular}} } 
& \begin{tabular}[c]{@{}l@{}} Cost \end{tabular} 
& -  & 1.71 & 1.70  & 1.59  & 1.57  \\ 
& \begin{tabular}[c]{@{}l@{}} Diff. \end{tabular}
& -  & -7.0  & -7.6  & -13.6 & -14.7   \\ \cline{2-7}
& \begin{tabular}[c]{@{}l@{}} {\tiny{$P_{\rm{LS}}$}} \end{tabular} 
& -  & 0  & 7.02  & 0   & 5.16  \\ \cline{2-7}
& \begin{tabular}[c]{@{}l@{}} {\tiny{$D_{\rm{d}}$}} \end{tabular} 
& -  & 5.51  & 1.05  & 26.1   & 17.9  \\ \cline{2-7}
& \begin{tabular}[c]{@{}l@{}} $ t_c $ \end{tabular}
& - & 83.4   & 87.1   &  86.9  & 85.2  \\
& \begin{tabular}[c]{@{}l@{}} Iter. \end{tabular}
& -   & 5   &  5  & 4   &  6 \\ \hline \hline
\end{tabular}
        \footnotesize
        \tnote{1} unit: $10^4$ m.u., \;
        \tnote{2} unit: \%, \;
        \tnote{3} unit: MW, \;
        \tnote{4} unit: nm, \;
        \tnote{5} unit: s.
\end{threeparttable}
\vspace{-3ex}
\end{table}

\paragraph{Optimality and complexity}
Then, the verification of the optimality and complexity performance is shown in Table \ref{tab:operating_cost_comparison}. For simplicity, we combine PPA with feasibility-guaranteed relaxation which is also named as PPA. Two-side management includes load shedding and reconfiguration besides the former three adjustment methods (GS, ESMC, and PPA). From the third row, it can be observed that the problem in failure mode has more variables than that in normal mode. Because there are the variables of redundant switches in semi-island mode. In the normal mode, there is no feasible solution only with GS because it is without PPA and feasibility-guaranteed relaxation. In the second column of the normal mode, $t_c$ of the {optimal management algorithm} is less than that of LNBD, and the iteration number has an opposite result. Due to time decomposing of LNBD, the subproblem is divided into 10 problems that are solved in sequence. Thus, the computational time at each iteration is larger than that of {optimal management algorithm}. However, the iteration number is related to the complexity of the master problem. After time decomposition, the master problem has a lower complexity. {As load shedding and reduced distance decrease the power demand, the operating cost in failure mode is less than that in normal mode.}
Besides, 
the LNBD can obtain a near-optimal solution that has a similar operating cost. {However, the complexity of LNBD can be reduced significantly in failure mode. The computation time $t_c$ is reduced from $198$s and $85.2$s}.

\subsection{Case Study 2: Island Mode}
The position of island-mode faults is shown in Fig. \ref{fig:fault_scenario}. To verify the algorithms' performance with different configurations, the SPS has one MTG and two ATGs that locate in zone 1, 3, and 6, and three ESMs are in zone 1, 3, and 5. The operation time is set to $T=12$.
The detail parameters are similar to the first one.

\begin{figure}[htbp]
\centering
\vspace{-0ex}
\includegraphics[width=0.48 \textwidth]{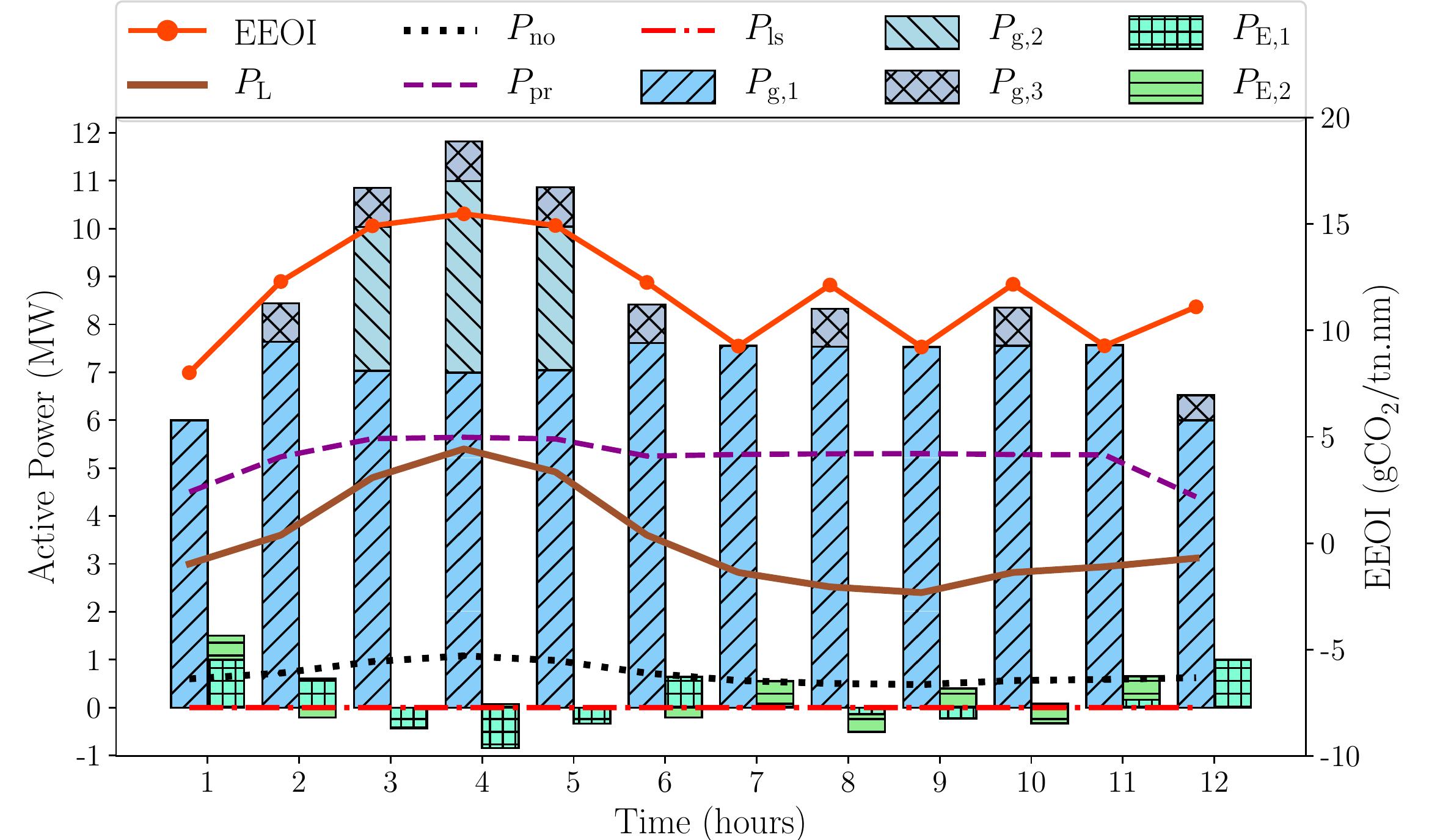}
\caption{  Power schedule with two-side management in island mode.}
\label{fig:power_scheduling_island}
\vspace{-2ex}
\end{figure}

\subsubsection{{Performance Analysis of Two-side Management}}
{The schedule with two-side adjustment methods when $D = 160$nm is shown in Fig. \ref{fig:power_scheduling_island}.
The varying $P_{\rm{PR}}(t)$ verifies that PPA works.
The load-shedding amount $P_{\rm{ls}}$ equal to zero. 
The ATG in zone 3 $(P_{{\rm g}, 2}(t))$ only works at the 3-rd, 4-th and 5-th time intervals due to the high load demand. 
Thus, GS, ESMC, and PPA are enough to finish the voyage distance $D = 160$nm.}

\begin{table}[ht]
\vspace{-2ex}
\centering
\caption{{Two-side management with different distances.}}
\label{tab:different_distance_island}
\begin{threeparttable}
\begin{tabular}{lllllll}
\hline \hline
\multirow{3}{*}{} & & \multicolumn{5}{c}{Distance $D$ (nm)} \\ \cline{3-7} 
                  & &  160 & 180   &  200  & 220 &  240  \\ \hline
\multirow{1}{*}{{\begin{tabular}[c]{@{}c@{}}Method in \cite{kanellos2014optimal}\end{tabular}} } 
& Cost \tnote{1} & 1.92 & -\tnote{2}  & - & - &  -  \\ 
\hline 
\multirow{3}{*}{{\begin{tabular}[c]{@{}c@{}}  Optimal manage- \\  ment algorithm\end{tabular}}} & Cost \tnote{1} & 1.89 & 2.38  &  2.60 &  2.60 &  2.60  \\ 
& $P_{\rm{LS}}$ (MW) & 0 & 0.26 & 6.99 & 6.99 & 6.99  \\
& $D_{\rm{d}}$ (nm) &  0 &  0    &  6.1   & 26.1 &  46.1 \\ \hline \hline
\end{tabular}
        \footnotesize
        \tnote{1} Cost unit: $10^4$ m.u. \;
        \tnote{2} - represents no feasible solution.
\end{threeparttable}
\vspace{-3ex}
\end{table}

{
The proposed coordination mechanism and feasibility-guaranteed mechanism are verified by the OPMSF problem with different travel distance $D$ compared with the method in [10]. The method in [10] is a two-side management method including GS, ESMC, and PPA.
It can be observed that load shedding always works when $D \geqslant 180$nm. Thus, the algorithm in \cite{kanellos2014optimal} cannot obtain a feasible solution when $D \geqslant 180$nm since it does not support load shedding and feasibility-guaranteed mechanism. Then, based on the results of last three columns, it can be obtained that the maximum travel distance $193.9$nm and the maximum amount of non-vital loads to reduce the fault effects is $6.99$ MW. 
Thus, the coordination mechanism and feasibility-guaranteed mechanism are also useful for the OPMSF problem in island mode. }

\subsubsection{Performance Analysis of Algorithms in island mode}

{Here the proposed algorithms are tested under the island mode. To verify the performance, $D=220$, and $\varphi(t) = 0.5$.}

\begin{table}[ht] \scriptsize  
\centering
\vspace{-3ex}
\caption{ Performance comparison between different algorithms in island mode.}
\label{tab:operating_cost_comparison_island}
\begin{threeparttable}
\begin{tabular}{cc|ccc|cc}
\hline \hline
\multicolumn{2}{c|}{Status} 
& \multicolumn{3}{c|}{Normal} & \multicolumn{2}{c}{\begin{tabular}[c]{@{}c@{}} Failure  \end{tabular}}\\ \cline{1-7}
\multicolumn{2}{c|}{\begin{tabular}[c]{@{}c@{}} Adjustment \\ methods \end{tabular} } 
& \begin{tabular}[c]{@{}c@{}}  Only \\ GS \end{tabular} & \begin{tabular}[c]{@{}c@{}} w/ ESMC, \\ GS \& PPA \end{tabular}
& \begin{tabular}[c]{@{}c@{}} Two-side \\ mgmt. \end{tabular}  
& \begin{tabular}[c]{@{}c@{}} w/ ESMC, \\ GS \& PPA  \end{tabular}
&  \begin{tabular}[c]{@{}c@{}} Two-side \\ mgmt. \end{tabular} \\ \hline
\multicolumn{2}{c|}{\begin{tabular}[c]{@{}c@{}} No. Var. (Con- \\ tinuous/Integer) \end{tabular}} 
& \begin{tabular}[c]{@{}c@{}} 108 \\ (36/72) \end{tabular} 
& \begin{tabular}[c]{@{}c@{}} 157 \\ (85/72) \end{tabular}
& \begin{tabular}[c]{@{}c@{}} 169 \\ (97/72) \end{tabular} 
& \begin{tabular}[c]{@{}c@{}} 157 \\ (85/72) \end{tabular}
&  \begin{tabular}[c]{@{}c@{}} 181 \\ (109/72) \end{tabular} \\ \hline
\multirow{6}{*}{\rotatebox{90}{\begin{tabular}[c]{@{}c@{}} Optimal manage-\\  ment algorithm \end{tabular}} } 
& \begin{tabular}[c]{@{}l@{}} Cost\tnote{1} \end{tabular}
& -  & 3.05 & 3.05  & 2.60  & 2.60  \\ 
& \begin{tabular}[c]{@{}l@{}} Diff.\tnote{2} \end{tabular}
& -  & \bf{0}   & 0   & -14.8  & -14.8  \\ \cline{2-7}
& \begin{tabular}[c]{@{}l@{}} {\tiny{$P_{\rm{LS}}$}}\tnote{3} \end{tabular} 
& -  & 0  & 8.39  & 0   & 6.99  \\ \cline{2-7}
& \begin{tabular}[c]{@{}l@{}} {\tiny{$D_{\rm{d}}$}}\tnote{4} \end{tabular} 
& -  & 8.69 & 4.64  & 30.3  & 26.1  \\ \cline{2-7}
& \begin{tabular}[c]{@{}l@{}} $ t_c $\tnote{5} \end{tabular}
& -  & *   &  *  & 125  &  137 \\
& \begin{tabular}[c]{@{}l@{}} Iter. \end{tabular}
& -  & *   &  *  & 87   &  95 \\ \hline
\multirow{6}{*}{\rotatebox{90}{\begin{tabular}[c]{@{}c@{}} LNBD \end{tabular}} } 
& \begin{tabular}[c]{@{}l@{}} Cost \end{tabular}
& -  & 2.95   & 2.97   &  2.69  & 2.70  \\ 
& \begin{tabular}[c]{@{}l@{}} Diff. \end{tabular}
& -  & -3.3  & -2.6  & -11.8  & -11.4   \\ \cline{2-7}
& \begin{tabular}[c]{@{}l@{}} {\tiny{$P_{\rm{LS}}$}} \end{tabular} 
& -  & 0  & 8.02  & 0   & 6.79  \\ \cline{2-7}
& \begin{tabular}[c]{@{}l@{}} {\tiny{$D_{\rm{d}}$}} \end{tabular} 
& -  & 9.17  & 5.20  & 31.5  & 26.1  \\ \cline{2-7}
& \begin{tabular}[c]{@{}l@{}} $ t_c $ \end{tabular}
& -  & 28.6  & 39.3  &  30.9  & 35.6  \\
& \begin{tabular}[c]{@{}l@{}} Iter. \end{tabular}
& -  & 4   &  4  & 5  &  6 \\ \hline \hline
\end{tabular}
        \footnotesize
        \tnote{1} unit: $10^4$ m.u., \;
        \tnote{2} unit: \%, \;
        \tnote{3} unit: MW, \;
        \tnote{4} unit: nm, \;
        \tnote{5} unit: s. \\
        \tnote{6} computation time is more than 1 hours. 
\end{threeparttable}
\vspace{-3ex}
\end{table}

The verification of the optimality and complexity performance is shown in Table \ref{tab:operating_cost_comparison_island}. 
From the third row, it can be observed that the OPMSF problem has more variables than that in normal mode. However, the increased variable is less than that in semi-island mode, because there are not redundant switches in island mode. $t_c$ of the {optimal management algorithm} in normal mode is more than one hour due to the high complexity of the master problem ($2^{72}$).
Different with Table \ref{tab:operating_cost_comparison}, although there are more continuous variables in failure mode, $t_c$ of the {optimal management algorithm} in failure mode is less than that in normal mode. Because the original problem is divided into two independent smaller problems and there are no coupled redundant switches between two island parts. LNBD in island mode has a faster computational time and a less iteration. Because in this case, the problem in normal mode is the more complex one. {In this scenario, load shedding reduces the demand of service loads, and reduced distance $D_{\rm d}$ decreases the demand of propulsion modules. Thus, the operating cost in failure mode is less than that in normal mode.}

Due to the sub-optimal power allocation in (\ref{eqn:P_E_online}) and (\ref{eqn:online_distance}), LNBD does not suit for the scenarios that the generator state has a significant change, and then $P_{{\rm{E}},w}(t)$ and $P_{\rm{PR}}(t)$ have a big difference with the change rule of service loads. 
$D_{\rm{d}}$ with same adjustment methods in failure mode is always larger than that in normal mode. It shows that the faults in above scenarios affect power supply-demand relationship and then reduce the maximum travel distance. It also verifies the necessity of post-fault management of SPS.


\section{Conclusion}
\label{sec:conclusion}
{In this paper, the OPMSF problem in a mid-time scale was investigated.
Firstly, a coordination mechanism was developed to make load shedding collaboratively work with GS and ESMC. In this mechanism, 
a sufficient condition of the penalty parameter in the load shedding term of the objective was derived to guarantee that load shedding only works when GS and ESMC cannot solve the OPMSF problem.
Then, considering the infeasible scenarios caused by faults, a feasibility-guarantee mechanism was established by adding a penalty term in the objective and deriving a sufficient condition of its penalty parameter. The  mechanism was to guarantee that if the original problem is feasible, the reformulated one has the same optimal solution; if not, the maximum travel distance can be obtained to assist rescue mission.
Finally, an optimal management algorithm based on BD and LNBD were designed to solve the reformulated problem. A complexity analysis was given to compare their performance.
The simulation demonstrated the effectivity of the mechanisms and algorithms, and the optimal management algorithm is suitable for solving the small-scale OPMS/OPMSF problem while LNBD for the large-scale one. 

The previous works and this paper all focused on the optimization problem at a fixed operation time and a determined route. However, the operation time and route have a great effect on the performance. Thus, how to construct the problem with variable route and operation time and design the effective algorithm to solve it requests further investigation.
}


{
\begin{appendices}

\section{Proof of Proposition 1}
\label{append:theorem1}
To meet the load demand before load shedding, the cost of load shedding $\Delta P_{\rm{no}}(t)$ has to be larger than the corresponding reduced cost of generation power:
\begin{align}
\label{eqn:load_shedding_proof1} \xi_{\rm{l}} \Delta P_{\rm{no}}(t) & = \xi_{\rm{l}} C_L ({\hat \rho}(t)) - \xi_{\rm{l}} C_L ({\rho}(t)) \nonumber \\
& > C_{\rm{G}} (\hat P_{{\rm{g}},m}(t) ) - C_{\rm{G}} (P_{{\rm{g}},m}(t) ), \\
\label{eqn:load_shedding_proof2} \xi_{\rm{l}} \Delta P_{\rm{no}}(t) & > \xi_{\rm{e}} (C_E (\hat P_{{\rm{e}}, n}(t) ) - C_E (P_{{\rm{e}}, n}(t) )).
\end{align}
where $\hat P_{\cdot}(t)$ and $P_{\cdot}(t)$ denote two different arbitrary output power at time $t$.

The (\ref{eqn:load_shedding_proof1}) and (\ref{eqn:load_shedding_proof2}) can be simplified as:
\begin{align}
\label{eqn:xi_form1} \xi_{\rm{l}}  & > \dfrac{a_{{\rm{g}},m}(\hat P_{{\rm{g}},m}(t) + P_{{\rm{g}},m}(t) ) \Delta P_{\rm{no}}(t) + b_{{\rm{g}},m} \Delta P_{\rm{no}}(t) }{\Delta P_{\rm{no}}(t)} \nonumber \\
& = a_{{\rm{g}},m}(\hat P_{{\rm{g}},m}(t) + P_{{\rm{g}},m}(t) ) + b_{{\rm{g}},m}, \\
\label{eqn:xi_form2} \xi_{\rm{l}}  & > \xi_{\rm{e}} a_{\rm{lc}} \dfrac{(\hat P_{{\rm{e}}, n}(t) + P_{{\rm{e}}, n}(t) ) \Delta P_{\rm{no}}(t)}{\Delta P_{\rm{no}}(t)} \nonumber \\
& = \xi_{\rm{e}} a_{\rm{lc}} (\hat P_{{\rm{e}}, n}(t) + P_{{\rm{e}}, n}(t) ).
\end{align}

The right terms of (\ref{eqn:xi_form1}) and (\ref{eqn:xi_form2}) can be simplified as:
\begin{align}
\label{eqn:xi_simp_form1} & a_{{\rm{g}},m}(\hat P_{{\rm{g}},m}(t) + P_{{\rm{g}},m}(t) ) + b_{{\rm{g}},m} 
< 2 a_{g} P_{g}^{\max} + b_{g}, \\
\label{eqn:xi_simp_form2} & \xi_{\rm{e}} a_{\rm{lc}} (\hat P_{{\rm{e}}, n}(t) + P_{{\rm{e}}, n}(t) )
< 2 \xi_{\rm{e}} a_{\rm{lc}} P_{\rm{e}}^{\max},
\end{align}
where $ P_{g}^{\max} = \underset{m} \max \left( P_{{\rm{g}},m}^{\max} \right), a_{g} = \underset{m} \max \left( a_{{\rm{g}},m} \right), b_{g} = \underset{m} \max \left( b_{{\rm{g}},m} \right), m \in \mathcal{M}$, $P_{\rm{e}}^{\max} = \underset{m} \max \left( P_{{\rm{e}}, n}^{\max} \right), n \in \mathcal{N} $.

Combining (\ref{eqn:xi_form1})-(\ref{eqn:xi_simp_form2}), we have that if $\xi_{\rm{l}}$ satisfies:
\begin{align}
\label{eqn:xi_simp_form3} \xi_{\rm{l}} & > 2 a_{g} P_{g}^{\max} + b_{g}, \\
\label{eqn:xi_simp_form4} \xi_{\rm{l}} & > 2 \xi_{\rm{e}} a_{\rm{lc}} P_{\rm{e}}^{\max},
\end{align}
the Eqs. (\ref{eqn:load_shedding_proof1}) and (\ref{eqn:load_shedding_proof2}) can be hold. They can be simplified as:
\begin{equation}
\label{eqn:xi_simp_form5} \xi_{\rm{l}} >  \max \{ 2 a_{g} P_{g}^{\max} + b_{g}, 2 \xi_{\rm{e}} a_{\rm{lc}} P_{\rm{e}}^{\max} \}.
\end{equation}

\section{Proof of Theorem 1}
\label{append:theorem2}
In order to prove that the relaxation is exact, it is necessary to show that any optimal solution of \textbf{P3} has equality in (\ref{eqn:distance_ppr_relax}). 
One optimal solution is denoted by ${\bm u}_w^*(t)$ which is expressed as 
\begin{align*}
{\bm u}_w^*(t)  {=} \big( & \bm \delta_{{\rm{g}},m}^*(t), \bm y_{{\rm{g}},m}^{*}(t), \bm{S}_{{\rm{P}}, x}^*, \bm{S}_{{\rm{S}}, y}^*, \bm P_{{\rm{g}},m}^*(t), \bm P_{{\rm{e}}, n}^*(t),  \\
& \bm P_{{\rm{pr}},r}^*(t), \rho_w^{*}(t) \big).
\end{align*}
For the sake of contradiction, we assume that ${\bm u}_w^*(t)$ has strict inequality, i.e.,
\begin{equation*}
\begin{aligned}
\label{eqn:distance_ppr_inequal}
& {D} - \sum\nolimits_{t \in \mathcal{T}} ( P_{{\rm{PR}},w}^*{(t)} / \alpha )^{1/\beta} \Delta t < 0.
\end{aligned}
\end{equation*}

Then, another solution $\tilde{\bm u}_w(t)$ is considered, which is defined by:
\begin{align*}
 & \tilde{\bm \delta}_{{\rm{g}},m}(t) =  \bm \delta_{{\rm{g}},m}^*(t), \tilde{\bm y}_{{\rm{g}},m}(t) = \bm y_{{\rm{g}},m}^{*}(t), \ \tilde{\bm{S}}_{{\rm{P}}, x} = \bm{S}_{{\rm{P}}, x}^*, \\
 & \tilde{\bm{S}}_{{\rm{S}}, y} = \bm{S}_{{\rm{S}}, y}^*, \tilde{\bm P}_{{\rm{g}},m}(t) = \bm P_{{\rm{g}},m}^*(t) - \varepsilon, \ \tilde{\bm P}_{{\rm{e}}, n}(t) = \bm P_{{\rm{e}}, n}^*(t), \\
 & \tilde{\bm P}_{{\rm{pr}},r}(t) = \bm P_{{\rm{pr}},r}^*(t) - \varepsilon, \tilde{\rho}_w (t) = \rho_w^{*}(t),
\end{align*}
where it satisfies that $ 0 < \varepsilon \leqslant P_{{\rm{pr}},r}^*(t) - \alpha(D/T \Delta t)^{1 / \beta} $.
It can be verified that $\tilde{\bm u}_w(t)$ satisfies all the constraints of \textbf{P3}, thus it is a feasible point. However, since $ \tilde{\bm P}_{{\rm{g}},m}(t) = \bm P_{{\rm{g}},m}(t)^* - \varepsilon $, the objective value of $\tilde{\bm u}_{w}(t)$ is strictly smaller than of ${\bm u}_{w}(t)^*$. This contradicts the assumption that ${\bm u}_{w}(t)^*$ is the optimal solution.

\section{Proof of Proposition 2}
\label{append:theorem3}
To guarantee that the optimal solution of \textbf{P4} is also the optimal solution of \textbf{P3} if \textbf{P3} has feasible solutions,
the increased cost $h D_{\rm{d}}$ must be greater than the decreased cost $\Delta C(t)$ of generation power caused by the reduced distance, which can
\begin{equation}
\begin{aligned}
\label{eqn:optimal_condition}
h D_{\rm{d}} & > \Delta C(t). \\
\end{aligned}
\end{equation}

According to (\ref{eqn:distance_ppr}) and Theorem \ref{th:theorem 2}, the left term of (\ref{eqn:optimal_condition}) can be transformed into
\begin{equation}
\begin{aligned}
\label{eqn:dist_cost_relaxation_1}
h D_{\rm{d}} = & h \sum_{t \in \mathcal{T} } \left( \hat P_{\rm{PR}}(t) /\alpha \right)^{\frac{1}{\beta}} \Delta t - h \sum_{t \in \mathcal{T} } \left( P_{\rm{PR}}(t)/\alpha \right)^{\frac{1}{\beta}} \Delta t, \\ 
\end{aligned}
\end{equation}
where $\hat P_{\rm{PR}}(t)$ and $ P_{\rm{PR}}(t)$ are the original and reduced power of propulsion modules at time $t$, respectively.

Then, based on Lagrange mean value theorem, it has
\begin{equation}
\begin{aligned}
\label{eqn:dist_cost_relaxation_2}
& \sum_{t \in \mathcal{T} } \left( \hat P_{\rm{PR}}(t) /\alpha \right)^{\frac{1}{\beta}} \Delta t - \sum_{t \in \mathcal{T} } \left( P_{\rm{PR}}(t)/\alpha \right)^{\frac{1}{\beta}} \Delta t \\ 
= & \dfrac{1}{\beta \alpha^{\frac{1}{\beta}}} \left( P_{\rm{PR}}^{'}(t)\right)^{\frac{1}{\beta}-1} \sum_{t \in \mathcal{T} } \left(\hat P_{\rm{PR}}(t) - P_{\rm{PR}}(t) \right) \Delta t, \\
\end{aligned}
\end{equation}
where $P_{\rm{PR}}^{'}(t) \in ( P_{\rm{PR}}(t), \hat P_{\rm{PR}}(t))$.

Consider that $\beta > 1$, there hold 
\begin{equation}
\label{eqn:moni}
( P_{\rm{PR}}^{'}(t))^{\frac{1}{\beta}-1} > ( P_{\rm{PR}}^{\max} )^{\frac{1}{\beta}-1} .
\end{equation}

Combing (\ref{eqn:dist_cost_relaxation_2}) and (\ref{eqn:moni}), (\ref{eqn:dist_cost_relaxation_1}) can be transformed into 
\begin{equation}
\begin{aligned}
\label{eqn:dist_cost_relaxation}
& h D_{\rm{d}} > \dfrac{h}{\beta \alpha^{1/\beta}} ( P_{\rm{PR}}^{\max} )^{\frac{1}{\beta}-1} \sum_{t \in \mathcal{T} }  (\hat P_{\rm{PR}}(t) - P_{\rm{PR}}(t)) \Delta t \\ 
= & \dfrac{h}{\beta \alpha^{1/\beta}} ( P_{\rm{PR}}^{\max} )^{\frac{1}{\beta}-1} \sum_{t \in \mathcal{T} } \big( \hat P_{\rm{G}}(t) + \hat P_{\rm{E}}(t) + \hat {\rho}(t) P_{\rm{no}}(t) \bigg.\\
& \bigg. - P_{\rm{G}}(t) -  P_{\rm{E}}(t) - {\rho}(t) P_{\rm{no}}(t) \big) \Delta t \\
= & \dfrac{h}{\beta \alpha^{1/\beta}} ( P_{\rm{PR}}^{\max} )^{\frac{1}{\beta}-1} \sum_{t \in \mathcal{T} } (\Delta P_{\rm{G}}(t) + \Delta P_{\rm{E}}(t) + \Delta P_{\rm{no}}(t)) \Delta t,
\end{aligned}
\end{equation}
where $\hat P_{\rm{PR}}(t)$ and $ P_{\rm{PR}}(t)$ are the reduced power of generators and ESMs, respectively.


At a deviation $\Delta P_{\rm{G}}(t)$, $\Delta P_{\rm{E}}(t)$, and $\Delta P_{\rm{no}}(t)$, the change of $\delta_{{\rm{g}},m}(t)$ gives a additional chance to reduce the operating cost. Hence, the best solution $\bm u_{w}(t)$ with a changed $\delta_{{\rm{g}},m}(t)$ has a lower operating cost than the one with same $\delta_{{\rm{g}},m}(t)$. Based on this, it can obtain that
\begin{equation*}
\begin{aligned}
\label{eqn:gen_cost_relaxation_1}
 \Delta C(t)
& \leqslant \sum_{t \in \mathcal{T}} \Big( \sum_{m \in \mathcal{M}} C_{\rm{G}}(\hat P_{{\rm{g}},m}(t) ) - \sum_{m \in \mathcal{M}} C_{\rm{G}}( P_{{\rm{g}},m}(t) ) \Big) \Delta t \\
& + \xi_{\rm{e}} \sum_{t \in \mathcal{T}} \Big( \sum_{n \in \mathcal{N}} C_{\rm{E}}( \hat P_{{\rm{e}}, n}(t) ) - \sum_{n \in \mathcal{N}} C_{\rm{E}}( P_{{\rm{e}}, n}({\rho}t) ) \Big) \Delta t \\
& + \xi_{\rm{l}} \sum_{t \in \mathcal{T}} \Big( C_{\rm{L}}( \hat {\rho}(t) ) -  C_{\rm{L}}( {\rho}(t) ) \Big) \Delta t. 
\end{aligned}
\end{equation*}

For simplicity, the above inequation can be transformed into
\begin{equation}
\begin{aligned}
\label{eqn:gen_cost_relaxation_2}
\Delta C(t)
< & \sum_{t \in \mathcal{T}} \sum_{m \in \mathcal{M}} (2 a_{{\rm{g}},m}P_{{\rm{g}},m}^{\max} + b_{{\rm{g}},m})  \Delta P_{{\rm{g}},m}(t) \Delta t\\
+ & 2 a_{\rm{lc}} \xi_{\rm{e}} \sum_{t \in \mathcal{T}} \sum_{n \in \mathcal{N}} P_{\rm{e}}^{\max}  \Delta P_{{\rm{e}}, n}(t) \Delta t \\
+ & \xi_{\rm{l}} \sum_{t \in \mathcal{T}} \sum_{n \in \mathcal{N}}  \Delta P_{\rm{no}}(t) \Delta t.
\end{aligned}
\end{equation}

Considering that $a_{\rm{lc}}$ in life-cycle cost function is a small constant, it has $(2 a_{{\rm{g}},m}P_{{\rm{g}},m}^{\max} + b_{{\rm{g}},m}) < (2 a_{g} P_{g}^{\max} + b_{g}) < \xi_{\rm{l}}$ and $\xi_{\rm{l}} > 2 \xi_{\rm{e}} a_{\rm{lc}} P_{\rm{e}}^{\max} $. Hence, there holds 
\begin{equation}
\begin{aligned}
\label{eqn:gen_cost_relaxation_3}
\Delta C(t)< & (2 a_{g} P_{g}^{\max} + b_{g}) \sum_{t \in \mathcal{T}} \sum_{m \in \mathcal{M}}  \Delta P_{{\rm{g}},m}(t) \Delta t \\
+ & 2 \xi_{\rm{e}} a_{\rm{lc}} P_{\rm{e}}^{\max} \sum_{t \in \mathcal{T}} \sum_{n \in \mathcal{N}}  \Delta P_{{\rm{e}}, n}(t) \Delta t \\
+ & \xi_{\rm{l}} \sum_{t \in \mathcal{T}}  \Delta P_{\rm{no}}(t) \Delta t \\
< & \xi_{\rm{l}} \sum_{t \in \mathcal{T}} \left( \Delta P_{\rm{G}}(t) + \Delta P_{\rm{E}}(t) + \Delta P_{\rm{no}}(t) \right) \Delta t.
\end{aligned}
\end{equation}

Thus, combing (\ref{eqn:dist_cost_relaxation}) and (\ref{eqn:gen_cost_relaxation_3}), it has that 
\begin{equation}
\begin{aligned}
\label{eqn:optimal_condition_h}
h   > \dfrac {\beta \alpha^{\frac{1}{\beta}} \xi_{\rm{l}} } {\left( P_{\rm{PR}}^{\max} \right)^{\frac{1}{\beta}-1} }.
\end{aligned}
\end{equation}

Therefore, if $h$ satisfies (\ref{eqn:optimal_condition_h}), there hold (\ref{eqn:optimal_condition}).
Consequently, the optimal solution of \textbf{P4} is also the optimal solution of \textbf{P3} if \textbf{P3} has feasible solutions and $h$ satisfies (\ref{eqn:optimal_condition_h}).

If \textbf{P3} has no feasible solution, $D_{\rm{d}}$ in any feasible solution must be greater than zero. (\ref{eqn:relaxed_cons_P5}) can rewritten as 
\begin{subequations}
\begin{align}
& {D}-D_{\rm{d}} - \sum\nolimits_{t \in \mathcal{T}} ( P_{{\rm{PR}},w}{(t)} / \alpha )^{1/\beta} \Delta t \\
\label{eqn:proof_cons_P5}
= & {D^{-}} - \sum\nolimits_{t \in \mathcal{T}} ( P_{{\rm{PR}},w}{(t)} / \alpha )^{1/\beta} \Delta t \leqslant 0.
\end{align}
\end{subequations}

Since there is no $D_{\rm{d}}$ in \textbf{P4} that replaced (\ref{eqn:relaxed_cons_P5}) with (\ref{eqn:proof_cons_P5}), $h D_{\rm{d}}$ in the objective of \textbf{P4} can be removed. Hence, this reformulated problem is similar with \textbf{P3}. The optimal solution $\bm u(t)^*$ of \textbf{P4} is also the optimal solution of \textbf{P3} that replaced (\ref{eqn:relaxed_cons_P5}) with (\ref{eqn:proof_cons_P5}). Based on the \textbf{Theorem \ref{th:theorem 2}}, it can be known that the maximum travel distance that can be achieved in time $T$ is $D^- = D- D_{\rm{d}}^*$.
\end{appendices}
}

{
\section*{Acknowledgment}
The authors would like to thank the anonymous reviewers for their professional and valuable comments that have led to the improved version.
}

%

\end{document}